\documentclass[a4paper,UKenglish]{lipics}

\usepackage{microtype}
\usepackage{nicefrac}

\usepackage{mathtools, amsmath, amssymb}
\usepackage[linesnumbered,boxed,noend]{algorithm2e}
\usepackage{color}
\usepackage{ifthen}
\usepackage{tikz}
\usepgflibrary{shapes}
\usetikzlibrary{arrows}
\usetikzlibrary{positioning}
\usepackage{graphicx}
\usepackage{multibib}

\newcites{add}{Additional References}

\tikzstyle{distr}=[inner sep=0mm, minimum size=1mm, draw, circle, fill]
\tikzstyle{state}=[inner sep=0.3mm, minimum size=3mm, draw]
\tikzstyle{trarr}=[semithick, -latex, rounded corners]

\tikzstyle{bad state}=[draw,diamond, inner sep=0mm, minimum size=7mm]
\tikzstyle{good state}=[draw, rectangle, inner sep=0mm, minimum size=5mm]
\tikzstyle{stochastic state}=[draw,circle, inner sep=0mm, minimum size=5mm]
\tikzstyle{target state}=[draw,circle,fill=lightgray, inner sep=0mm, minimum size=5mm]
\tikzstyle{dead state}=[draw,circle, inner sep=0mm, minimum size=5mm]
\tikzstyle{arr}=[-latex', rounded corners]

\newcommand{\arginf}{\operatorname*{arg\,inf}}
\newcommand{\argsup}{\operatorname*{arg\,sup}}
\newcommand{\mydef}{\stackrel{\mbox{\rm {\tiny def}}}{=}}
\newcommand{\startpara}[1]{{%
\vskip5pt\noindent
{\bf #1.}}}

\newcommand{\mdp}{M}

\newcommand{\states}{S}
\newcommand{\actions}{A}
\newcommand{\trans}{T}
\newcommand{\gain}{F}
\newcommand{\interest}{\rho}

\newcommand{\aWork}{\mathit{work}}
\newcommand{\aInvest}{\mathit{invest}}
\newcommand{\aProfit}{\mathit{profit}}
\newcommand{\aLoss}{\mathit{loss}}

\newcommand{\disc}[1]{\mathit{disc}(#1)} 

\newcommand{\calF}{\mathcal{F}}

\newcommand{\Runs}{\mathsf{Runs}}
\newcommand{\Hist}{\mathsf{Hist}}
\newcommand{\fpat}{w}

\newcommand{\Cone}{\mathsf{Cone}}

\newcommand{\Rset}{\mathbb{R}}

\newcommand{\discn}[2]{\mathit{disc}_{#2}(#1)}

\newcommand{\size}[1]{||#1||}

\newcommand{\calL}{\mathcal{L}}
\newcommand{\Succ}{\mathit{Succ}}

\newcommand{\Zset}{\mathbb{Z}}
\newcommand{\tot}{\mathit{tot}}
\newcommand{\NP}{\mathit{NP}}
\newcommand{\coNP}{\mathit{coNP}}

\newcommand{\Val}[2][XXX]{\ifthenelse{\equal{#1}{XXX}}{\mathit{Val}(#2)}{\mathit{Val}_{#1}(#2)}}
\newcommand{\WR}[2][XXX]{\ifthenelse{\equal{#1}{XXX}}{\vec{W}(#2)}{\vec{W}_{#1}(#2)}}
\newcommand{\WRvec}{\vec{W}}
\newcommand{\eps}{\varepsilon}

\newcommand{\sinit}{s_{0}}
\newcommand{\xinit}{x_{0}}

\newcommand{\dist}{\mathit{dist}}
\newcommand{\St}[1][XXX]{\ifthenelse{\equal{#1}{XXX}}{\Sigma}{\Sigma_{#1}}}

\newcommand{\run}[1]{\mathit{Run}(#1)}
\newcommand{\gamerun}[1]{\mathit{run}(#1)} 
\newcommand{\game}{G}
\newcommand{\Prb}{\mathbb{P}}
\renewcommand{\Pr}[3]{\Prb^{#1}_{#2,#3}}   

\newcommand{\Qset}{\mathbb{Q}}
\newcommand{\Nset}{\mathbb{N}}

\newcommand{\Win}{\mathit{Win}}

\newcommand{\calO}{\mathcal{O}}
\newcommand{\Hit}{\mathit{Hit}}
\newcommand{\Cover}[1]{\mathit{Ar}(#1)}
\newcommand{\kernel}{\mathit{ker}}

\newcommand{\LU}{\mathit{LU}}

\newcommand{\norm}[1]{\size{#1}_{\infty}}
\newcommand{\theoremlike}[2]{\par\medskip\penalty-250%
{{\bfseries\noindent
#2 \ref{#1}.}}\it}

\newcommand{\thmhelperpre}[2]{\theoremlike{#1}{#2}}
\newcommand{\thmhelperpost}{\par\medskip}

\newenvironment{reftheorem}[1]{\thmhelperpre{#1}{Theorem}}{\thmhelperpost}
\newenvironment{reflemma}[1]{\thmhelperpre{#1}{Lemma}}{\thmhelperpost}

\newenvironment{refproposition}[1]{\thmhelperpre{#1}{Proposition}}{\thmhelperpost }

\renewcommand{\rho}{\varrho}
\newcommand{\discount}{\beta}
\newcommand{\dmdp}{D}
\newcommand{\thr}{Thr}
\renewcommand{\vec}[1]{\mathbf{#1}}
\newcommand{\vertices}{V}

\theoremstyle{plain}
\newtheorem{proposition}[theorem]{Proposition}

\begin{document}

\title{Solvency Markov Decision Processes with Interest}
\author[1]{Tom\'{a}\v{s} Br\'{a}zdil\footnote{Tom\'{a}\v{s} Br\'azdil is supported by the Czech Science Foundation, grant No P202/12/P612.}}
\author[2]{Taolue Chen}
\author[2]{Vojt\v{e}ch Forejt\footnote{ Vojt\v{e}ch Forejt is also affiliated with FI MU, Brno, Czech Republic.}}
\author[1]{Petr Novotn\'y}
\author[2]{Aistis Simaitis}
\affil[1]{Faculty of Informatics, Masaryk University, Czech Republic}
\affil[2]{Department of Computer Science, University of Oxford, UK}


\maketitle

\begin{abstract}

Solvency games, introduced by Berger et al., provide an abstract framework
for modelling decisions of a risk-averse investor, whose goal is to avoid ever going broke.
We study a new variant of this model, where, in addition to stochastic environment and
fixed increments and decrements to the investor's wealth, we introduce
interest, which is earned or paid on the current level of savings or debt, respectively.

We study problems related to the minimum initial wealth sufficient to avoid bankrupt\-cy (i.e. steady decrease of the wealth) with probability at least $p$.
We present an exponential time algorithm which approximates this minimum
initial wealth, and show that a polynomial time approximation is not possible unless P $=$ NP.
For the qualitative case, i.e. $p=1$, we show that the problem whether a given number is larger than or equal to the minimum initial wealth belongs to NP $\cap$ coNP, and show that a polynomial time algorithm would yield a polynomial time algorithm for mean-payoff games,
existence of which is a longstanding open problem.
We also identify some classes of solvency MDPs for which this problem is in P.
In all above cases the algorithms also give corresponding bankruptcy avoiding strategies.

\end{abstract}

\section{Introduction}

Markov decision processes (MDP) are a standard model of complex decision-making where results of decisions may be random. An MDP has a set of {\em states}, where each state is assigned a set of {\em enabled actions}. Every action determines a distribution on the set of successor states. A run starts in a state; in every step, a {\em controller} chooses an enabled action and the process moves to a new state chosen randomly according to the distribution assigned to the action. The functions that describe decisions of the controller are called \emph{strategies}. They may depend on the whole history of the computation and the choice of actions may be randomized.

MDPs form a natural model of decision-making in the financial world.
To model nuances of financial markets, various MDP-based models have been developed (see~e.g.~\cite{Schal:MDP-finance,BU:book,BKSV:Solvency-games}).
A common property of these models is that actions correspond to investment choices and result in (typically random) payoffs for the controller.
One of the common aims in this area is to find a {\em risk-averse} controller (investor) who strives to avoid undesirable events~\cite{hamilton1994time,hull2009options}.

In this paper we consider a model based on standard reward structures for MDPs, which is closely related to solvency games studied in~\cite{BKSV:Solvency-games}. The model is designed so that it captures essential properties of risk-averse investments. We assume finite-state MDPs and assign a (real) reward to every action which is collected whenever the action is chosen. The states of the MDP capture the global situation on the market, prices of assets, etc. 
Note that it is usually plausible to model the prices by a finite-state stochastic process 
(see~e.g.~\cite{Schal:MDP-finance}).
Rewards model
money received (positive rewards) and money spent
 (negative rewards) by the controller.
Controllers are then compared w.r.t. their ability to collect the reward over finite or infinite runs.

Standard objectives such as the {\em total reward}, or the {\em long-run average reward} are not suitable for modelling the behaviour of a risk-averse investor as they allow temporary loss of an arbitrary amount of money (i.e., a long sequence of negative rewards), which is undesirable, because normally the controller's access to credit is limited.
The authors of~\cite{BKSV:Solvency-games} consider a ``bankruptcy-avoiding'' objective defined as follows: Starting with an initial amount of wealth $W_0$,
in the $n$-th step, the current wealth $W_n$ is computed from $W_{n-1}$ by adding the reward collected in the $n$-th step. The goal is to find a controller which maximizes the probability of having $W_n>0$ for all $n$.

Although the model of~\cite{BKSV:Solvency-games} captures basic behaviour of a risk-averse investor, it lacks one crucial aspect usually present in the financial
environment, i.e., the {\em interest}. Interests model the value that is received from
holding a certain amount of cash, or conversely, the
cost of having a negative balance.
To accommodate interests, we propose the following extension of the bankruptcy-avoiding objective: Fix an interest rate $\varrho>1$.\footnote{%
For notational convenience, we define the interest rate to be the number $1+r$, where $r>0$ is the usual interest rate, i.e. the percentage of money paid/received over a unit of time.
}
Starting with an initial wealth $W_0$, in the $n$-th step, compute the current wealth $W_n$ from $W_{n-1}$ by adding not only the collected reward but also the interest $(\varrho-1) W_{n-1}$.
The economical motivation for such a model is that the controller
can earn additional amount of wealth by lending its assets for a fixed interest,
and conversely, when the controller is in debt,
it has to pay interest to its creditors (for the clarity of presentation, we suppose the interest earned from positive wealth is the same as the interest paid on debts).

Hence, the objective is to ``manage'' the wealth
so that it stays above some threshold and
does not keep decreasing to negative infinity. More precisely, we want to maximize the probability of having $\liminf_{n\rightarrow \infty} W_n>-\infty$. %
Intuitively, $\liminf_{n\rightarrow \infty} W_n\geq 0$ means that
the controller ultimately does not need to borrow money, and
$-\infty<\liminf_{n\rightarrow \infty} W_n<0$ means that the controller is able to sustain interest payments from
its income.
If $\liminf_{n\rightarrow \infty} W_n=-\infty$, then the
controller cannot sustain interest payments and bankrupts.

An important observation is that this objective is closely related to
another well-studied objective concerning the {\em discounted total reward}.
Concretely, given a discount factor $0<\beta<1$, the discounted total reward $T$ accumulated on a run is defined to be the weighted sum of rewards of all actions on the run where 
the weight of the $n$-th action is $\beta^{n}$. In particular, the \emph{threshold problem} asks to maximize the probability of $T\geq t$ for a given threshold
 $t$. This problem has been 
considered in, e.g.,~\cite{Sob82,CS87,White93,WL99}. 
A variant of the threshold problem is the \emph{value-at-risk} problem \cite{BodaF06} which asks, for a given probability $p$, 
what is the infimum threshold, such that maximal probability of discounted reward surpassing the threshold is at least $p$? 
We show that for every controller, the probability of $T\geq t$ with discount factor $\beta$ is equal to the probability of $\liminf_{n\rightarrow \infty} W_n>-\infty$ with $W_0=-t$
for the interest rate $\varrho\mathop{:=}\frac{1}{\beta}$.
This effectively shows interreducibility of these problems. Note that the interpretation of the discount factor as the inverse of the interest is natural in financial mathematics.

\smallskip

\noindent
\textbf{Contribution.}
We introduce a model of solvency MDPs with interests (referred to as {\em solvency MDPs} for brevity), which allows
to capture the complex dynamics of wealth management under uncertainty.
We show that for every solvency MDP there
is a bound on wealth such that above this bound the
bankruptcy is surely avoided (no matter what the controller is doing), and another bound on wealth
below which the bankruptcy is inevitable.
Nevertheless, we also show that there still might be infinitely many reachable values of wealth between these two bounds.

The main results of our paper concentrate on the complexity of computing minimal wealth with which the controller can stay away from bankruptcy. 
Let $\WR{s_0,p}$ be the {\em infimum} of all initial wealths $W_0$ such that starting in the state $s_0$ with $W_0$ the controller can avoid bankruptcy (i.e., $\liminf_{n\rightarrow \infty} W_n>-\infty$)
with probability at least $p$.
Our overall goal is to compute this number $\WR{s_0,p}$. Solution to this problem is important for a risk-averse investor, 
whose aim is to keep the risk of bankruptcy below some %
acceptable level.

First we consider the \emph{qualitative case}, i.e. $\WR{s_0,1}$. For this case we show a connection with two-player (non-stochastic) games with discounted total reward objectives.
Then, using the results of~\cite{ZP:games-graphs} we show that
there is an {\em oblivious} strategy (i.e., the one that looks only at the current state but is independent of
the wealth accumulated so far) which starting in some state $s_0$ with wealth $\WR{s_0,1}$ avoids bankruptcy with probability one.
The problem whether $W\geq \WR{s_0,1}$ for a given $W$ (encoded in binary) is in $\NP\cap\coNP$ (we also obtain a reduction from discounted total reward games,
showing that improving this complexity bound might be difficult). In addition, the number $\WR{s_0,1}$ can be computed in pseudo-polynomial time.
Further it follows that for a restricted class of solvency Markov chains (i.e. when there is only one enabled action in every state)
the value $\WR{s_0,1}$ can be computed in polynomial time.

The main part of our paper concerns the {\em quantitative case}, i.e. $\WR{s_0,p}$ for an arbitrary probability bound $p$.
\begin{itemize}
\item We give an exponential-time algorithm that approximates  $\WR{s_0,p}$ up to a given absolute error $\eps>0$.
 We actually show that the algorithm runs in time polynomial in the number of control states and exponential in
 $\log(1/(\rho-1))$, $\log(1/\eps)$ and  $\log(r_{\max})$, where $\rho$ is the interest rate and $r_{\max}$ is the maximal $|r|$ where $r$ is a reward associated to some action.
\item Employing a reduction from the Knapsack problem, we show that the above complexity cannot
 be lowered to polynomial in either $\log(1/\eps)$ or $\log(\rho-1)+ \log(r_{\max})$ unless P$=$NP.
\item We give an exponential-time algorithm that for a given $\varepsilon>0$ and initial wealth $W_0$ computes $v$ such that 
if the initial
wealth is increased by $\varepsilon$, then
the probability of avoiding bankruptcy is at least $v$ (i.e.
 $W_0+\varepsilon\geq \WR{s_0,v}$) and $v\geq \sup\{v'\mid W_0\in \WR{s_0,v'}\}$.
\end{itemize}
Moreover, via the aforementioned interreducibility between discounted and solvency MDPs we establish new complexity bounds
for value-at-risk approximation in discounted MDPs.

We note that the aforementioned algorithms employ a careful rounding of numbers representing the current wealth $W_n$. Choosing the right precision for this rounding is quite an intricate step, since a naive choice would only yield a doubly-exponential algorithm.

The paper is organized as follows:
after introducing necessary definitions and clarifying the relation with the discounted MDPs in Section~\ref{sec:prelims}
we summarise the results for qualitative problem in Section~\ref{sec:qualitative}.
In Section~\ref{sec:quantitative}
we give the contributions for the quantitative problem.

\smallskip
\noindent
\textbf{Related work.}
Processes involving interests and their formal models naturally emerge in the field of financial mathematics. An MDP-based model of a financial market is presented, e.g., in Chapter 3 of \cite{BU:book}. There, in every step the investor has to allocate his current wealth between riskless bonds, on which he receives an interest according to some fixed interest rate, and several risky stocks, whose price is subject to random fluctuations. Optimization of the investor's portfolio with respect to various utility measures was studied. However, this portfolio optimization problem was considered only in the finite-horizon case, where the trading stops 
after some fixed number of steps. In contrast, we concentrate on the long-term stability of the investor's wealth. Also, the model in \cite{BU:book} was analysed mainly from the mathematical perspective (e.g., characterizing the form of optimal portfolios), while we focus on an efficient algorithmic computation of the optimal investor's behaviour.

The issues of a long-term stability and algorithms were considered for other related models,
all of which concern 
total accumulated reward properties.
Our model is especially close to
{\em solvency games}~\cite{BKSV:Solvency-games},
which are in fact MDPs with a single control state, where the investor aims to keep the total accumulated reward non-negative.
In {\em energy games}~(see e.g.~\cite{CdAHS03,CHD:energy-games,CHD:energy-MDPs}),
there are two competing players, but no stochastic behaviour.
In \emph{one-counter} MDPs~\cite{BBEK:OC-games-termination-approx-journal},
the counter can be seen as a storage for the current
value of wealth. All these models differ from the topic
studied in this paper in that they do not consider interest on wealth.
This makes them fundamentally different in terms of their properties,
e.g. in our setting the set of all wealths reachable from a given initial
wealth can have nontrivial limit points.
Also, in all the three aforementioned models,
the objective is to stay in the positive wealth.
Here we focus on a different objective to capture the idea that
it is admissible to be in debt as long as it is possible
to maintain the debt above some limit.

As mentioned before, our work is also related to the threshold discounted total reward  objectives, 
which were considered in \cite{Sob82,CS87,White93,WL99}, %
where the authors studied finite- and infinite-horizon
cases.
In the finite-horizon case, in particular \cite{WL99} gave an algorithm to compute the probability, but a careful analysis shows that their algorithm has a doubly-exponential worst-case complexity when the planning horizon (i.e., the number of steps after which the process halts) is encoded in binary. In \cite{BodaFLS04} they proposed to approximate the probability through the discretisation of wealth, but in the worst the error of approximation is $1$, no matter how small discretisation step is taken.
In \cite{WL99}, the optimality equation characterising optimal probabilities has been provided for the infinite-horizon case, but
no algorithm was proposed. 
Moreover, \cite{BodaF06} considered the ``value-at-risk'' problem, but again only for the finite-horizon
case, giving a doubly-exponential approximation algorithm.
Although we consider only infinite-horizon MDPs, the
exponential-time upper bound for the
$\WR{s,p}$ approximation and the NP-hardness lower bound can be easily carried over to the finite-horizon case. Thus, we establish new complexity bounds for value-at-risk approximation in both finite and infinite-horizon discounted MDPs. %
We also mention \cite{FKR95} which introduced the percentile performance criteria where 
the controller aims to find a strategy achieving a specified value of the long-run limit average reward 
at a specified probability level (percentile). %

\section{Preliminaries}
\label{sec:prelims}

We denote by $\Nset$, $\Zset$, $\Qset$ and $\Rset$ the sets of all natural, integer, rational and real numbers, respectively.
For an index set $I$, its member $i$ and vector $\vec{V}\in \Rset^I$ we denote by $\vec{V}(i)$ the $i$-component of $\vec{V}$. 
The encoding size of an object $B$ is denoted by $\size{B}$. We use $\log x$ to refer to the binary logarithm of $x$. We assume that all numbers are represented in binary and that rational numbers are 
represented as fractions of binary-encoded integers.

We assume familiarity with basic notions of probability theory.
Given an at most countable set $X$, we use $\dist(X)$ to denote all
probability distributions on $X$.

\begin{definition}[MDP]
A \emph{Markov decision process} (MDP) is a tuple
$\mdp = (\vertices, \actions, \trans)$ where $\vertices$ is at most countable set of \emph{vertices},
$\actions$ is a finite set of \emph{actions}, and
$\trans:\vertices\times \actions \rightarrow \dist(\vertices)$ is a partial \emph{transition function}. We assume that for every $v \in V$ the set $\actions(v)$ of all actions available at $v$ (i.e., the set off all actions $a$ s.t. $\trans(v,a)$ is defined) is nonempty.
\end{definition}

We denote by $\Succ(v,a)=\{u\mid \trans(v,a)(u)>0\}$ the \emph{support} of $\trans(v,a)$.
A {\em Markov chain} is an MDP with 
one action per vertex, i.e., $|\actions(v)|=1$
for all $v \in \vertices$.

From a given initial vertex $v_0\in\vertices$ the MDP evolves as follows.
An {\em infinite path} (or {\em run}) is a sequence $v_0a_1v_1a_2v_2\dots\in (\vertices\times \actions)^\omega$ 
such that $a_{i+1}\in \actions(v_i)$ and $v_{i+1}\in \Succ(v_i,a_{i+1})$ for all $i$.
A {\em finite path} (or \emph{history}) is a prefix
of a run ending with a vertex, i.e. a word of the form  $(\vertices\times \actions)^*\vertices$.
We refer to the set of all runs as $\Runs_\mdp$ and to the set of all histories as $\Hist_{\mdp}$.
For a finite or infinite path $\omega=v_0a_1v_1a_2v_2\dots$ and $i \in \Nset$ we denote by $\omega_i$ the 
finite path
$v_0a_1\cdots a_{i}v_i$.

A \emph{strategy} in $\mdp$ is a function that to every history $w$
assigns a distribution on actions available in the last vertex of $w$.
A strategy is {\em deterministic} if it always assigns distributions that choose some action with probability $1$, and {\em memoryless}
if it only depends on the last vertex of history. %
We use $\St[\mdp]$ (or just $\St $) for the set of all strategies of $\mdp$.

Each history $\fpat\in \Hist_{\mdp}$ determines the set $\Cone(\fpat)$ consisting of
all runs having $\fpat$ as a prefix. To an MDP $\mdp$, its vertex $v$ and strategy $\sigma$
we associate the probability space
$(\Runs_\mdp,\calF,\Pr{\sigma}{\mdp}{v})$, where $\calF$ is the $\sigma$-field generated by all $\Cone(\fpat)$,
and $\Pr{\sigma}{\mdp}{v}$ is the unique probability measure such that for every history $w=v_0a_1\dots a_{k}v_k$ we have
$\Pr{\sigma}{\mdp}{v}(\Cone(w)) =
\mu(v_0) \cdot \prod_{i=1}^{k} x_i$, where $\mu(v_0)$ is $1$ if $v_0=v$ and $0$ otherwise, and where
$x_i = \sigma(w_{i-1})(a_i)\cdot T(v_{i-1},a_i)(v_i)$ for all \mbox{$1 \leq i \leq k$} (the empty
product is equal to $1$). We drop $M$ from the subscript when the MDP is clear from the context.

\begin{definition}[Solvency MDP]
A \emph{solvency Markov decision process}
is a tuple $(\states, \actions, \trans, \gain, \interest)$ where $\states$ is a finite set of \emph{states}, $A$ and $T$ are such that $(\states, \actions, \trans)$ is an MDP,
$\gain: \states\times \actions \rightarrow \Qset$ is a partial \emph{gain function}  and $\interest\in\Qset \cap (1,\infty)$ is an \emph{interest rate}.
\end{definition}
We stipulate that for every $(s,a)\in \states\times\actions$ the value $\gain(s,a)$ is defined iff $a\in \actions(s)$. 
A {\em solvency Markov chain} is a solvency MDP with one action per state, i.e. $|\actions(s)|=1$
for all $s \in \states$.
A \emph{configuration} of a solvency MDP $\mdp=(\states, \actions, \trans, \gain,\interest)$
is represented as a state-wealth pair $(s,x)$ where $s\in \states$ and $x\in \Qset$.
The semantics of $\mdp$ is given by an infinite-state MDP
$\mdp_\interest=(\states\times \Qset, \actions, \trans_\interest)$ where
for every $(s,x)\in \states\times\Qset$ and $a\in \actions(s)$ we
define $\trans_\interest((s,x),a)(s',\interest\cdot x + \gain(s,a)) = p$ whenever
$\trans(s,a)(s') = p$.
We sometimes do not distinguish between $\mdp$ and $\mdp_\interest$ and refer
to strategies or runs of $\mdp$ where strategies or runs of $\mdp_\interest$
are intended.
A strategy $\sigma$ for $\mdp_\interest$ is {\em oblivious} if it is memoryless and
does not make its decision based on the current wealth, i.e. for all $w\cdot (s,x)$ and
$(s,x')$ we have $\sigma(w\cdot (s,x)) = \sigma((s,x'))$.

\startpara{Objectives}
 Given an solvency MDP $\mdp$ and its initial configuration $(s_0, x_0)$, we are interested in the set of runs in which
 the wealth always stays above some finite bound, denoted by
 $\Win = \Runs_\mdp \setminus \{(s_0, x_0)a_1(s_1, x_1)\cdots \in \Runs_\mdp \mid
 \liminf_{n\rightarrow \infty} x_n = -\infty\}$.
 Intuitively, this objective models the ability of the
investor not to go bankrupt, i.e. to compensate for the incurred interest by 
obtaining sufficient gains.
 We denote $\Val[\mdp]{s_0,x_0} = \sup_{\sigma} \Pr{\sigma}{\mdp}{(s_0,x_0)}(\Win)$
 the maximal probability of winning with a given wealth, and
 $\WR[\mdp]{s,p} = \inf \{x \mid \Val[\mdp]{s,x} \ge p\}$ the infimum of wealth sufficient for winning with probability
 $p$.
In this paper we are mainly interested in the problems of computing or approximating
the values of $\WR[\mdp]{s,p}$.
We also address the problem of computing a convenient
risk-averse strategy for an investor with a given initial wealth $x_0$. 
A precise definition of what we mean by a convenient strategy is 
given in Section \ref{sec:quantitative} (Theorem \ref{thm:quantitative-algorithm}).
We say that a strategy is $p$-winning (in an initial configuration $(s_0,x_0)$)
if $\Prb^{\sigma}_{\mdp,(s_0,x_0)}(\Win)\ge p$. A $1$-winning strategy is called
{\em almost surely winning}, and strategy $\sigma$ with  $\Prb^{\sigma}_{\mdp,(s_0,x_0)}(\Win)= 0$ is called {\em almost surely losing}. %

\begin{example}\label{ex:running}
 Consider the following solvency MDP $\mdp=(\states, \actions, \trans, \gain, \interest)$:
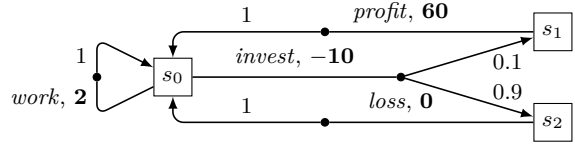
\begin{center}
\begin{tikzpicture}
\tikzstyle{every node}=[font=\footnotesize]
 \node[good state] at (0,0) (s0) {$s_0$};
 \node[distr] at (3,0) (d0a) {};
 \node[distr] at (-1,0) (d0b) {};
 \node[good state] at (5,0.6) (s1) {$s_1$};
 \node[good state] at (5,-0.6) (s2) {$s_2$};
 \node[distr] at (2,0.6) (d1) {};
 \node[distr] at (2,-0.6) (d2) {};

 \draw[trarr] (s0) -- (d0a) node[midway, above] {$\aInvest$, $\mathbf{-10}$} -- (s1) node[pos=0.8, below] {$0.1$};
 \draw[trarr] (d0a) -- (s2) node[pos=0.8, above] {$0.9$};
 \draw[trarr] (s1) -- (d1) node[pos=0.65, above] {$\aProfit$, $\mathbf{60}$} -| (s0) node[pos=0.25, above] {$1$} ;
 \draw[trarr] (s2) -- (d2) node[pos=0.65, above] {$\aLoss$, $\mathbf{0}$} -| (s0) node[pos=0.25, above] {$1$};
 \draw[trarr] (s0) -- +(-1,-0.5) -- (d0b) node[midway, left] {$\aWork$, $\mathbf{2}$} -- +(0,0.5) node[midway, left] {$1$} -- (s0);
 
\end{tikzpicture}
\end{center}
 Here $\states=\{s_0,s_1,s_2\}$,
 $\actions=\{\aWork, \aInvest, \aProfit, \aLoss\}$,
 $\trans$ is depicted by the arrows in the figure, for example $\trans(s_0,\aInvest) = [s_1\mapsto 0.1, s_2 \mapsto 0.9]$, the function $\gain$ is given by the
 bold numbers next to the actions, e.g. $\gain(s,\aWork)=2$, and $\varrho=2$ (we take this extremely large value to keep the example computations simpler).
 The MDP models the choices of a person who can either work,
 which ensures certain but relatively small income, 
 or can invest a larger amount of money but take a significant risk.
 Starting in the configuration $(s_0,-10)$ (i.e. in debt), an example strategy $\sigma$ is the strategy which always chooses $\aWork$ in $s_0$,
 but as can be easily seen, we get $\Pr{\sigma}{\mdp}{(s_0,-10)}(\Win) = 0$ since the constant gains are not high enough to cover the interest incurred by the debt.
 An optimal strategy here is to pick $\aWork$ only in histories ending with a configuration $(s_0,x)$ for $x\ge -2$, and to pick $\aInvest$ otherwise.
 Such strategy shows that $\Val[\mdp]{s_0,-10} = 0.1$. Now suppose that the investor wants to find out what is the wealth needed to make sure the probability of
 winning is at least $0.7$, i.e. wants to compute $\WR[\mdp]{s_0,0.7}$.
 This number is equal to $-2$. To see this, observe that for any configuration $(s_0,y)$ where $y<-2$ the optimal strategy must pick $\aInvest$, which with probability
 $0.9$ results in a debt from which it is impossible to recover. Finally, observe that $\Val[\mdp]{s_0,-2}=1$ since a strategy that always chooses $\aWork$ is $1$-winning in $(s_0,-2)$. This demonstrates that the function $\Val{s,\cdot}$ for a given state $s$ may not be continuous.
\end{example}

\noindent
\textbf{Relationship with discounted MDPs.}
The problems we study for solvency MDPs are closely related to another risk-averse decision making model, so called \emph{discounted MDPs with threshold objectives}.
A discounted MDP is a tuple $\dmdp=(\states,\actions,\trans,\gain,\discount)$, where the first four components are as in a solvency MDP and $0 < \beta < 1$ is a \emph{discount factor}. The semantics of a discounted MDP is given by a finite-state MDP $\dmdp^\discount=(\states,\actions,\trans)$ and a reward function $\disc{\cdot}$ which to every run $\omega=s_0 a_1 s_1 a_2 \dots$ in $\dmdp^\discount$ assigns its \emph{total discounted reward} $\disc{\omega}=\sum_{i=1}^{\infty}\gain(s_{i-1},a_i)\cdot \beta^{i}$. The threshold objective asks the controller to maximize, for a given threshold $t \in \Qset$, 
the probability of the event $\thr(t)=\{\omega\in\run{\dmdp^{\discount}}\mid \disc{\omega}\geq t\}$.

Now consider a solvency MDP $\mdp=(\states, \actions, \trans, \gain, \interest)$ with an initial configuration $(s_0,x_0)$ and a discounted MDP $\dmdp=(\states, \actions, \trans, \gain, 1/\interest)$ with a threshold objective $\thr(-x_0)$.
Note that once an initial configuration $(s_0,x_0)\in \states\times\Qset$ is fixed, there is a natural one-to-one correspondence between runs in $\mdp_\interest$ initiated in $(s_0,x_0)$ and runs in $\dmdp^{1/\interest}$ initiated in $s$: we identify a run $(s_0,x_0)a_1(s_1,x_1)a_2 \dots$ in $\mdp_{\interest}$ with a run $s_0a_1s_1a_2\dots$ in $\dmdp^{1/\interest}$. This correspondence naturally extends to strategies in both MDPs, so we assume that these MDPs have identical sets of runs and strategies.

\begin{proposition}
\label{prop:discounted}
Let $\mdp$, $\dmdp$ be as above. Then $\Pr{\sigma}{\mdp}{(s,x)}(\Win) = \Pr{\sigma}{\dmdp}{s}(\thr(-x))$
for all $\sigma {\in} \St$.
\end{proposition}
\begin{proof}
It suffices to show that for every run $\omega$ we have $\omega\in\Win \Leftrightarrow \disc{\omega}\geq -x$.
Fix a run $\omega=(s_0,x_0)a_1(s_1,x_1)a_2\dots$, and define, for every $n\geq 0$,
$\discn{\omega}{n}\mydef\sum_{i=1}^{n}\gain(s_{i-1},a_{i})\cdot\frac{1}{\rho^{i}}$ (an empty sum is assumed to be equal to 0). Obviously, for every $n \geq 0$ we have $x_n = \rho^n\cdot(\discn{\omega}{n}+x_0)$. Thus, if $\disc{\omega}=\lim_{n \rightarrow \infty}\discn{\omega}{n}> -x_0$, then $\lim_{n\rightarrow \infty} x_n$ exists and it is equal to $+\infty$. Similarly, if $\disc{\omega}<-x_0$, then $\lim_{n\rightarrow \infty} x_n = - \infty$. If $\disc{\omega}=-x_0$, the infimum wealth $x_n$ along $\omega$ is finite (see Appendix~\ref{app:discounted}), and so $\omega\in \Win$.
\end{proof}
It follows that many natural problems for solvency MDPs (value computation etc.) are polynomially equivalent to similar natural problems for discounted MDPs with threshold objectives. 
In particular, our problem of computing/approximating $\WR[\mdp]{\sinit,p}$ is interreducible with the \emph{value-at-risk} problem in discounted MDPs, where the aim is to compute/approximate the supremum threshold $t$ such that under suitable strategy the probability (risk) of the 
discounted reward being $\leq t$ is at most $1-p$. 

\section{Qualitative Case}
\label{sec:qualitative}

In this section we establish a connection between the qualitative problem for solvency MDPs (i.e., determining whether $x \geq \WR[\mdp]{s,1}$ for a given state $s$ and number $x$) and
the problem of determining the winner in non-stochastic discounted games.

\begin{definition}[Discounted game]
A finite \emph{discounted game} is a tuple $\game = (S_1, S_2, s_0, T, R, \beta)$ where
 $S_1$ and $S_2$ are sets of player 1 and 2 states, respectively; $s_0\in S_1$
 is the initial state; $T\subseteq (S_1 \times  S_2) \cup (S_2\times S_1)$ is
 a transition relation; $R: (S_1\cup S_2) \rightarrow \Rset$ is a reward function; and
 $0<\beta<1$ is a discount factor.
\end{definition}

A strategy for player $i\in \{1,2\}$ in a discounted game is a function $\zeta_i:(S_1\cup S_2)^*\cdot S_i \rightarrow (S_1\cup S_2)$
such that $(s,\zeta_i(ws))\in T$ for every $s$ and $w$. A strategy is \emph{memoryless} if it only depends
on the last state.
A pair of strategies $\zeta_1$ and $\zeta_2$ for players 1 and 2 yields a unique run $\gamerun{\zeta_1,\zeta_2}=s_0s_1\ldots$ in the game,
given by $s_j = \zeta_i(s_0\ldots s_{j-1})$ where $i$ is $1$ or $2$ depending on whether $s_{j-1}\in S_1$ or
$s_{j-1}\in S_2$. The discounted total reward of the run is defined to be $\disc{s_0s_1\ldots} := \sum_{i=0}^\infty \beta^{i+1} R(s_i)$.
The \emph{discounted game} problem asks, given a game $\game$ and a value $x$, whether there is a strategy $\zeta_1$ for player 1
such that for all strategies $\zeta_2$ of player 2 we have $\disc{\gamerun{\zeta_1,\zeta_2}} \ge x$. Such a strategy $\zeta_1$ is then called \emph{winning}.

By Proposition~\ref{prop:discounted} the problem of determining whether $x\ge\WR[\mdp]{s,1}$ for a state $s$ of a solvency MDP $\mdp$ is
interreducible (in polynomial time)
with the problem of determining whether there is $\sigma\in \St[\dmdp]$ 
such that $\Pr{\sigma}{\dmdp}{s}(\thr(-x)) = 1$ in the corresponding discounted MDP $\dmdp$.
We show that
the latter is interreducible\footnote{Actually, we use slightly different variants of the discounted game problem in reductions \emph{from} and \emph{to} the discounted MDPs problem, respectively. Nevertheless, 
they establish the desired complexity bounds.} with the discounted game problem.

Let us first fix a discounted MDP
$\dmdp=(\states,\actions,\trans,\gain,\discount)$. We say that a run $\omega=s_0a_1s_1\ldots$ of $\dmdp$ is \emph{realisable}
under a strategy $\sigma$
if $\sigma(s_0a_1\ldots s_n)(a_{n+1})>0$, and $\trans(s_n,a_{n+1})(s_{n+1}) >0$ for all $n$.
The idea of the reduction relies on the following lemma, which is proved in Appendix~\ref{app:qualitative}.
\begin{lemma}\label{lem:all}
If $\sigma\in \St[\dmdp]$ satisfies $\Pr{\sigma}{\dmdp}{s}(\thr(x))=1$, 
then \emph{all} runs realisable under $\sigma$ are in~$\thr(x)$.
\end{lemma}
Using the lemma above we can construct a game $\game$ from $\dmdp$
by stipulating that the results of actions are chosen by player 2 instead of being chosen randomly, %
and vice versa. The technical details of the reduction are presented in  Appendix~\ref{app:qualitative}.
The next theorem follows from the reduction and the fact that memoryless (deterministic) strategies suffice in discounted games.
\begin{theorem}
\label{cor:md_strats}
For every solvency MDP $\mdp$ there exists an oblivious deterministic strategy which is almost-surely winning in every configuration $(s,x)$ with $x\geq \WR[\mdp]{s,1}$.
\end{theorem}
The discounted game problem is in $\NP\cap \coNP$ and there exists a pseudopolynomial algorithm computing the optimal value~\cite{ZP:games-graphs}. Also, when
one of the players controls no states in a game, the problem can be solved in polynomial 
time~\cite{ZP:games-graphs}. Hence, we get the following theorem.
\begin{theorem}
 The qualitative problem for solvency MDPs is in $\NP\cap\coNP$. Moreover, there is a pseudopolynomial algorithm that 
 computes $\WR[\mdp]{s,1}$ for every state $s$ of $\mdp$. For the restricted class of solvency Markov chains, to compute $\WR[\mdp]{s,1}$ and to decide the qualitative problem 
 can be done in polynomial time.
\end{theorem}

Note that the  existence of a reduction from mean-payoff games to discounted games~\cite{Andersson2009} suggests that improving the above complexity to polynomial-time is difficult, since a polynomial-time algorithm for solvency MDPs
would give a polynomial-time algorithm for mean-payoff games, existence of which is a longstanding open problem in the area of graph games.

\section{Quantitative Case}
\label{sec:quantitative}

This section formulates results on quantitative questions
for solvency MDPs.
We start with a proposition showing that we can restrict our attention to some subset of $\states\times\Qset$, since for every state there are
two values below and above which
all strategies are almost-surely winning or 
losing, respectively.
Intuitively, these values represent wealth (positive or negative)
for which losses/gains from the interest dominate gains/losses
from the gain function $\gain$.
An important consequence of the proposition, when combined with~\cite{Martin:Blackwell-determinacy}, is that deterministic strategies suffice to maximize
the probability of winning. Therefore, in the rest of this section we consider only deterministic strategies. The proposition is proved in Appendix~\ref{app:lgug}.

\begin{proposition}
\label{prop:lgug}
For every state $s$ of the solvency MDP $\mdp$ there are rational numbers
 \[
  U(\mdp,s) \mydef \arginf_{x\in \Rset} \forall \sigma \,.\,\mathbb{P}^{\sigma}_{\mdp,(s,x)}(\!\Win) = 1
  \quad\text{and}\quad
  L(\mdp,s) \mydef \argsup_{x\in \Rset} \forall \sigma \,.\,\mathbb{P}^{\sigma}_{\mdp,(s,x)}(\!\Win) = 0,
 \]
 of encoding size polynomial in $\size{\mdp}$, and they can be computed in polynomial time using linear programming techniques.
 Moreover, we have $\mathbb{P}^{\sigma}_{\mdp,(s,U(M,s))}(\Win) = 1$ for every strategy~$\sigma$.
\end{proposition}
To illustrate the proposition, we return to Example~\ref{ex:running} and note that $U(\mdp,s_0) = \frac{20}{3}$ and $L(\mdp, s_0) = -\frac{40}{3}$.
Obviously, for every $s$ we have $K \geq U(\mdp,s) \geq L(\mdp,s) \geq -K$ where $K = \max_{(s,a)\in S\times A}\frac{|F(s,a)|}{\rho-1}$, but as Example~\ref{ex:running} shows, using $U(\mdp,s)$ and $L(\mdp,s)$ we can restrict the set of interesting configurations more than with the trivial bounds $K$ and $-K$. 

We also define the global versions of the bounds, i.e.,
$L(\mdp) \mydef \min_{s\in\states} L(\mdp,s)$ and
$U(\mdp) \mydef \max_{s\in\states} U(\mdp,s)$.
In accordance with the economic interpretation of our model,
we call any configuration of the form $(s,x)$ with
$x\geq U(M,s)$ a \emph{rentier configuration}. From Proposition~\ref{prop:lgug} it follows that every run which visits a rentier configuration belongs to $\Win$.

Note that although Proposition~\ref{prop:lgug} suggests that we can restrict
our analysis to the configurations $(s,x)$ where $L(\mdp,s)\le x \le U(\mdp,s)$,
the set of reachable configurations  between these bounds
is still infinite in general as the following example shows.

\begin{example}\label{ex:infstate}
  Consider a solvency MDP $\mdp=(\{s\}, \{a,b\}, \trans, \gain, \frac{3}{2})$ with $\trans(s,a)=\trans(s,b)=s$,
  and $\gain(s,a)=\frac{1}{2}$ and $\gain(s,b)=-\frac{1}{2}$.
  We have $L(\mdp)=-1$ and $U(\mdp)=1$.
  We will show that for any $n \in \Nset$ there
  is a configuration $(s,x_n)$ where $x_n = k/2^n$
  that is reachable in exactly
  $n$ steps from an initial configuration $(s,\frac{1}{2})$ and satisfies
  $k \in \Nset_0$,
  $0\leq k < 2^n$,
  $2\nmid k$. 
  Hence the reachable state space from $(s,\frac{1}{2})$ is infinite
  as the numbers $x_n$ are pairwise different.

  We set $x_0 = \frac{1}{2}$, and let $(s,x_n)$ be a reachable configuration where $x_n$ is of the form $k/2^n$
  satisfying the above conditions. In one step we
  can reach configurations
  $(s,x')$ where $x'= \interest x_n \pm \frac{1}{2} = \frac{3k\pm 2^n}{2^{n+1}}$. Clearly $2 \nmid 3k \pm 2^n$; otherwise we would
  have $2 \mid 3k$ and thus $2 \mid k$ which contradicts the definition of $x_n$. It remains to
  show that one of the values of $x'$ again satisfies the above conditions; this is a simple exercise, and
  we give a proof in Appendix~\ref{app:infstate}.
\end{example}
  Note that if the interest $\rho$ is restricted to
  be an integer, the reachable configuration space between $L(\mdp)$ and $U(\mdp)$ is finite, 
   because for the initial configuration $(s,x)$ it holds
    $x=\frac{p}{q}$ where $p,q\in \Zset$,
   and $\interest \cdot x+y = \frac{\interest \cdot p+y\cdot q}{q}$. Hence, any reachable wealth
   is a multiple of $\frac{1}{q}$, and there are only finitely many such numbers between
   $L(\mdp)$ and $U(\mdp)$. This means that one can use off-the-shelf algorithms for finite-state MDPs, i.e., minimising the probability
   to reach configuration with $(s,x)$, where
   $x< L(\mdp,s)$. 
   However, for the general case, this is not possible and we need to devise new techniques.

\subsection{Approximation Algorithms}

In this subsection we show how to approximate $\WR{s,p}$.
Our algorithm depends
on the following theorem, which allows us, in~a certain sense that will be explained soon, to approximate the function 
$\Val[M]{s_0,\cdot}$.

\begin{theorem}
\label{thm:quantitative-algorithm}
 There is an algorithm that computes, for a solvency MDP $M$ with initial configuration $(s_0,x_0)$ and a given $\eps>0$, a rational number $v$ and a strategy $\sigma$ such that:
\begin{enumerate}
 \item $v\geq \Val[M]{s_0,x_0}$.
 \item %
 Strategy $\sigma$ is $v$-winning from configuration $(s_0,x_0+\eps)$.
\end{enumerate}
The running time of the algorithm is polynomial in
$|S|\cdot|A| \cdot\log\left(p_{\min}^{-1}\right)$
where
$p_{\min}=\min_{(s,s',a)\in S^2\times A}T(s,a)(s')$,
and exponential in $\log(|r_{\max}|/(\rho-1))$ and $\log(1/\eps)$
where $r_{\max}=\max_{(s,a)\in S\times A}|F(s,a)|$.
\end{theorem}

We will prove Theorem~\ref{thm:quantitative-algorithm} later, but first we argue that
the theorem is important in its own right.
Consider the following
scenario.
Suppose that an investor starts with wealth $x_0$. 
It is plausible to assume that this initial 
wealth is not strictly fixed.
Instead, one can assume that the investor is willing to acquire
some small additional amount of wealth (represented by $\eps$), 
in exchange for some substantial benefit. 
Here, the benefit consists of the fact that the 
small difference in the initial wealth allows
the investor to compute and execute a strategy, 
under which the risk of bankruptcy is provably
no greater than the lowest risk achievable with the original wealth. 
Note that the strategy $\sigma$ may not be $\Val[M]{s_0,x_0+\eps}$-winning 
from $(s_0,x_0+\eps)$. %
We now proceed with the theorem providing the approximation of $\WR{s,x}$

\begin{theorem}\label{thm:approx}
For a given solvency MDP $M$, its state $s$ and rational numbers $\delta>0$, $p\in[0,1]$, it is possible to approximate $\WR{s,p}$ up to the absolute error $\delta$ in time polynomial in $(|S|\cdot|A|)^{\calO(1)}\cdot\log\left(p_{\min}^{-1}\right)$, and exponential in $\log(|r_{\max}|/(\rho-1))$ and $\log(1/\delta)$, where $p_{\min}$ an $r_{\max}$ are as in Theorem~\ref{thm:quantitative-algorithm}.
\end{theorem}
\begin{proof}
Suppose that we already know that $a \leq \WR{s,p} \leq b$, for some $a,b$. We can use the algorithm of Theorem \ref{thm:quantitative-algorithm} for $s_0=s$, $x_0 =a+ (b-a)/2$ and $\eps=\mbox{$(b-a)/4$}$. If the algorithm returns $v\leq p$, we know that $a+(b-a)/2 \leq \WR{s,p} \leq b$, otherwise we can conclude that $a \leq \WR{s,p} \leq a+3(b-a)/4$. Initially we know that $L(M)\leq \WR{s,p} \leq U(M)$, so in order to approximate $\WR{s,p}$ with absolute error $\delta$ it suffices to perform $\calO(\log((U(M)-L(M))/\delta))$ iterations of this procedure, finishing when $\eps\leq\delta/4$.
\end{proof}
Later we will show that the time complexity of the algorithm cannot be improved to polynomial in either $\log(|r_{\max}|/(\rho-1))$ or $\log(1/\delta)$ unless P$=$NP.

\begin{proof}[Proof of Theorem \ref{thm:quantitative-algorithm}]

For the rest of this section we fix a solvency MDP $M=(S,A,T,F,\rho)$ and its initial configuration $(s_0,x_0)$.
First we establish the existence of a strategy that, given a small additional amount of wealth, reaches a rentier configuration in at most exponential number of steps with probability at least $\Val[M]{s_0,x_0}$. Then, we will show how to compute such a strategy in exponential time.

To establish the proof of the following proposition, we use a suitable Bellman functional whose unique fixed point is  equal to $\WRvec$. The proof can be found in Appendix~\ref{app:bellman-gen}. 

\begin{proposition}
\label{prop:finite-memory-reach}
For every initial configuration $(s,x)$ and every $\eps>0$ 
there is a strategy $\sigma_{\eps}$ such that starting in
$(s,x+\eps/2)$, $\sigma_{\eps}$ ensures hitting of a rentier 
configuration in at most 
$n=\big\lceil \frac{\log{(U(M)-L(M))}+\log{\eps^{-1}}+2}{\log \rho} \big\rceil$
steps with probability at least $\Val[M]{s,x}$. 
In particular, $\Prb^{\sigma_{\eps}}_{(s,x+\eps/2)}(\Win)\geq \Val[M]{s,x}$.
\end{proposition}

\noindent
The previous proposition shows that the number $v$ and strategy $\sigma$ of Theorem \ref{thm:quantitative-algorithm} can be computed by examining the possible behaviours of $M$ during the first $n$ steps. 
However, since $\log\rho \approx \rho - 1$ for $\rho$ close to $1$, the number $n$ can be exponential in $\size{M}$. 
Thus, the trivial algorithm, that unfolds the MDP from the initial configuration $(s_0,x_0+\eps/2)$ into a tree of depth $n$, and on this tree computes a strategy maximising the probability of reaching a rentier configuration, has a doubly-exponential complexity. 
The key idea allowing to reduce this complexity to singly-exponential is
to round the numbers representing the wealth in the configurations of $M$ to numbers of polynomial size. If the size
is chosen
carefully, the error introduced by the rounding is not large enough to thwart the computation.
In the following we assume that
$\log\rho< \log(U(M)-L(M))+\log(\eps^{-1})+2$,
since otherwise $n=1$ and we can compute the strategy $\sigma$ and number $v$ by computing an action that maximizes the one-step probability of reaching a rentier configuration from $(s_0,x_0+\eps/2)$.

We now formalise the notion of rounding the numbers appearing in configurations of $M$.
Let $\lambda$ be a rational number.
We say that two configurations $(s,x)$, $(s',x')$ are {\em $\lambda$-equivalent}, denoted by $(s,x)\sim_{\lambda}(s',x')$, if $s=s'$ and one of the following conditions holds:
\begin{itemize}
 \item both $x$ and $x'$ are greater than $U(M,s)$ or less than or equal to $L(M,s)$; or
 \item $L(M,s) < x,x' \le U(M,s)$ and there is $k\in \Zset$ such that both $x, x' \in (k\lambda,(k+1)\lambda]$.
\end{itemize}

Clearly, $\sim_{\lambda}$ is indeed an equivalence on the set $S\times \Qset$, and every member of the quotient set $(S\times \Qset) / \mathord{\sim}_{\lambda}$ is a tuple of the form $(s,D)$, with $s\in S$ and $D$ being either a half-open interval of length at most $\lambda$ or one of the intervals $(U(M,s),+\infty)$, $(-\infty,L(M,s)]$. For such $D$, we denote by $w_{D}$ the maximal element of $D$ (putting $w_{(U(M,s),+\infty)}=+\infty$). We also denote by $[s,x]_{\lambda}$ the equivalence class of $(s,x)$.

Now let $n$ be as in Proposition $\ref{prop:finite-memory-reach}$.
We define an MDP $M_{\lambda, n}$ representing an unfolding of $M$ into a DAG of depth $n$, in which the current wealth $w$ is always rounded up to the least integer multiple of $\lambda$ greater than $w$, with configurations exceeding the upper or dropping below the lower threshold of Proposition~\ref{prop:lgug} being immediately recognized as winning or losing. The unfolded MDP $M_{\lambda,n}$ is formally defined as follows.

\begin{definition}\label{def:unfold}[Unfolded MDP]
Let $\mdp=(\states, \actions, \trans, \gain, \interest)$
be an solvency MDP, and $n>0$ and $\lambda>0$ two numbers. We define an
MDP $M_{\lambda,n}=(S',A,T')$ where $S'$ is $((S\times\Qset)/\mathord{\sim}_{\lambda})\times\{0,1,\dots,n\}$,
and the transition function $T'$ is the unique function satisfying the following:
 \begin{itemize}
  \item for all $(s,D,i)\in S'$ and $a\in A$ where $i< n$ and $D$ is a bounded interval, the distribution $T'((s,D,i),a)$ is defined iff $a\in A(s)$, and assigns $T(s,a)(s')$ to $([s',\rho\cdot w_{D} + F(s,a)]_{\lambda},i+1)$
  \item for every other vertex $(s,D,i)\in S'$ there is only a self loop on this vertex under every action, i.e., $T'((s,D,i),a)$ is given by $[(s,D,i)\mapsto 1]$ for every action $a\in \actions$. %
 \end{itemize}
\end{definition}
The size of $M_{\lambda,n}$ as well as the time needed to construct it is
\mbox{$(|S|\cdot|A| \cdot\log(p_{\min}^{-1})\cdot n\cdot\lambda^{-1})^{\calO(1)}$}.

Now we denote by $\Hit$ the set of all runs in $M_{\lambda,n}$ that contain a vertex of the form $(t,(U(M,t),\infty),i)$, and by $\Cover{z}$ (for ``almost rentier'') the set of all runs in $M$ that hit a configuration of the form $(t,y)$ with $y\geq U(M,t)-z$ in at most $n$ steps. In particular, $\Cover{0}$ is the event of hitting a rentier configuration in at most $n$ steps. The following lemma (proved in Appendix~\ref{sec:discretization}) shows that $M_{\lambda,n}$ adequately approximates the behaviour of $M$.

\begin{lemma}
\label{lem:discretization}
Let $(s,y)$ be an arbitrary configuration of $M$. Then the following holds:
\begin{enumerate}
 \item For every $\sigma\in \St[M]$ there is $\pi\in \St[M_{\lambda,n}]$ such that
  $\Prb^{\pi}_{M_{\lambda,n},([s,y]_{\lambda},0)}(\Hit) \geq \Prb^{\sigma}_{M,(s,y)}(\Cover{0})$.
 \item There is $\sigma\in \St[M]$ such that
  $\Prb^{\sigma}_{M,(s,y)} (\Cover{n\cdot\lambda\cdot\rho^{n}}) \geq \sup_{\pi}\,\Prb^{\pi}_{M_{\lambda,n},([s,y]_{\lambda},0)} (\Hit) \mydef v$,
  where the supremum is taken over $\St[M_{\lambda,n}]$. Moreover, the number $v$ and a finite representation of the strategy $\sigma$ can be computed in time $\size{M_{\lambda,n}}^{\calO(1)}$.
\end{enumerate}
\end{lemma}

We can now finish the proof of Theorem \ref{thm:quantitative-algorithm}. Let us put
$\lambda=\lceil(64\cdot n\cdot(U(M)-L(M))^2)/\eps^3\rceil^{-1}$.
An easy computation (shown in Appendix~\ref{app:neps}) proves that $n \cdot \lambda \cdot \rho^{n}\leq \frac{\eps}{2}$ thanks to our assumption that $\log\rho< \log(U(M)-L(M))+\log(\eps^{-1})+2$.

\SetKwRepeat{DoWhile}{do}{while}
\begin{algorithm}[t]
\small
  \KwIn{MDP $\mdp$, state $s$, number $p\in [0,1]$, $\delta>0$}
  \KwOut{number and strategy satisfying conditions of Theorem~\ref{thm:approx}}
  $a:= L(\mdp,s)$; $b:=U(\mdp,s)$
   \tcp*{\footnotesize see Proposition~\ref{prop:lgug}}
  \DoWhile{$(b-a) /4> \delta$} 
  {
   $\varepsilon := (b-a)/4$; \quad
   $n :=\lceil (\log{(U(M)-L(M))}+\log{\eps^{-1}}+2)/\log \rho\rceil$\;
   $y := a+(b-a)/2$; \quad
   $\lambda :=\lceil (64\cdot n\cdot(U(M)-L(M))^2) / \eps^3 \rceil^{-1}$\;
   $\mdp' :=M_{\lambda,n}$ and initial configuration $t:=[s,y]_{\lambda}$
    \tcp*{\footnotesize see Definition~\ref{def:unfold}}
   $v:= \sup_\pi \Prb^{\pi}_{\mdp',t} (\Hit)$;
    \tcp*{\footnotesize $\Val[M]{s,y}\le v \le \Val[M]{s,y+\varepsilon}$ by Theorem~\ref{thm:quantitative-algorithm}}
   \lIf{$v \le p$}{$a:= a+(b-a)/2$} \lElse{$b:=a+3(b-a)/4$}\;
  }
    Compute $\sigma$ s.t. $\Prb^{\sigma}_{M,(s,y)} (\Cover{n\cdot\lambda\cdot\rho^{n}})\ge v$
    \tcp*{\footnotesize see Lemma~\ref{lem:discretization}}
  \Return $a$ and the wealth-independent strategy (see p.~\pageref{page:strat-shift}) given by $\sigma$ and $(s,y)$\;
\caption{Algorithm approximating $\WR{s,p}$}\label{alg:wr}
\end{algorithm}

By Proposition \ref{prop:finite-memory-reach} there is a strategy $\sigma_{\eps}$ in $M$ with $\Prb^{\sigma}_{M,(s_0,x_0+\eps/2)}(\Cover{0})\geq \Val[M]{s_0,x_0}$, and so from Lemma \ref{lem:discretization} (1.) we get $\sup_{\pi}\Prb^{\pi}_{M_{\lambda,n},([s_0,x_0+\eps/2]_{\lambda},0)} (\Hit) \geq \Val[M]{s_0,x_0}$. By part (2.) of the same lemma we can compute, in time $\size{M_{\lambda,n}}^{\calO(1)}$, a strategy $\sigma$ in $M$ and a number $v$ such that $\Prb^{\sigma}_{M,(s_0,x_0+\eps/2)} (\Cover{\eps/2}) \geq v \geq \Val[M]{s_0,x_0}$. In other words, from $(s_0,x_0+\eps/2)$ the strategy $\sigma$ reaches with probability at least $v$ a configuration that is only $\eps/2$ units of wealth away from being rentier.
Note that once an initial configuration is fixed, any strategy can be viewed as being
\label{page:strat-shift}
{\em wealth-independent}, i.e. being only a function of a sequence of states and actions in the history, since the current wealth can be inferred from this sequence and the initial wealth. Suppose now that we fix the initial 
configuration $(
s_0,
x_0+\eps)$ instead of $(s_0,x_0+\eps/2)$, keeping the same strategy~$\sigma$ (i.e., we use a strategy that selects the same action as $\sigma$ after observing the same sequence of states and actions). It is then obvious that we reach a rentier configuration with probability at least $v$, i.e., $
\Prb^{\sigma}_{(s,x+\eps)} (\Win) \geq v$ as required.%

It remains to analyse the complexity of the construction. The analysis is merely technical and
is postponed to Appendix~\ref{sec:alg-complexity}.
\qedhere{\em(Thm.~\ref{thm:quantitative-algorithm})}
\end{proof}
The results described in this section are summarised in a form of pseudocode in
Algorithm~\ref{alg:wr}.

\subsection{Lower Bounds}
Now we complement the positive results given above with lower complexity bounds. 

\begin{theorem}\label{thm:hard}
\label{thm:quant-hardness}
The problem of deciding whether $\WR{s,p} \leq x$ for a given $x$ is $\NP$-hard. Furthermore, existence of any of the following algorithms is not possible unless P$=$NP:
 \begin{enumerate}
  \item An algorithm approximating $\WR{s,p}$ up to the absolute error $\delta$ in time polynomial in $|S|\cdot|A|\cdot\log\left(p_{\min}^{-1}\right)$ and $\log(|r_{\max}|/(\rho-1))$ and exponential in $\log(1/\delta)$.
  \item An algorithm approximating $\WR{s,p}$ up to the absolute error $\delta$ in time polynomial in $|S|\cdot|A|\cdot\log\left(p_{\min}^{-1}\right)$ and $\log(1/\delta)$ and exponential in $\log(|r_{\max}|/(\rho-1))$.
 \end{enumerate}
Above, the numbers $r_{\max}$ and $p_{\min}$ are as in Theorem~\ref{thm:quantitative-algorithm}.
\end{theorem}
\begin{proof}[Proof sketch]
 We show how to construct, for a given instance of the Knapsack problem, a solvency MDP $M$ in which the item values are suitably encoded into probabilities of certain transitions, while the item weights are encoded as rewards associated to some actions. We then show that the instance of Knapsack has a solution if and only if for a certain state $s$ of $M$ and a certain number $p$ (which can be computed from the instance) it holds that $\WR{s,p}\leq 0$. We also show that in order to decide this inequality it suffices (for the constructed MDP $M$) to approximate $\WR{s,p}$ up to the absolute error $\frac{1}{4}$. (Intuitively, this corresponds to the well-known fact that no polynomial approximation algorithm for Knapsack can achieve a constant absolute error.) To get part (2.) we use a slight modification of the same approach.
 
 Let us note that a crucial component of the aforementioned reductions is that $\Val[M]{t,\cdot}$ may not be a continuous function (see example~\ref{ex:running}). Intuitively, this allows us to recognise whether the current wealth, which in $M$ always encodes weight of some set of items, surpasses some threshold. The proof can be found in Appendix~\ref{app:hardness}.
\end{proof}
Note that thanks to the interreducibility from Proposition~\ref{prop:discounted}, the (suitably rephrased) results of Theorems~\ref{thm:approx} and~\ref{thm:quant-hardness} hold also for the value-at-risk approximation in discounted MDPs.

\section{Conclusions}
We have introduced solvency MDPs, a model
apt for analysis of systems
where interest is paid 
or received for the accumulated wealth.
We have analysed the complexity
of fundamental problems, and proposed algorithms
that approximate 
the minimum wealth needed to win 
with a given probability and compute a strategy that achieves the goal. 
As a by-product, we obtained new results for the \emph{value-at-risk} problem in discounted MDPs.

There are several important directions of future study. 
One question deserving attention is to find an algorithm computing
or approximating $\Val{s,x}$. 
The usual approaches of discretising the state 
space do not work in this case since the function $\Val{s,\cdot}$
is not continuous and thus it 
is difficult to  bound the error introduced by the discretisation. 
Another direction is the implementation
of the algorithms and their evaluation on case-studies.

\bibliographystyle{abbrv}
\bibliography{references}

\newpage
\appendix
\noindent
\begin{center}
{\huge{\bf Technical Appendix}}
\end{center}
\section{Proofs for Section~\ref{sec:qualitative}}
\label{app:qualitative}

\subsection{Proof of Proposition~\ref{prop:discounted}}
\label{app:discounted}
We show the missing part of the proof of Proposition~\ref{prop:discounted}. We need to show that the infimum wealth $x_n$ along $\omega$ is finite.
This follows from the fact that for every 
$n\geq 0$ we have
\begin{align*}
 x_n &= \rho^n\cdot(\discn{\omega}{n}-\disc{\omega})=-\rho^n\cdot(\sum_{i=n+1}^{\infty}\frac{1}{\rho^i}\cdot\gain(s_{i-1},a_{i})) \\ &\geq -\frac{\rho^n}{\rho^{n+1}}\cdot\frac{\max_{(s,a)\in\states\times\actions}|\gain(s,a)|}{1-\frac{1}{\rho}}=-\frac{\max_{(s,a)\in\states\times\actions}|\gain(s,a)|}{\rho-1}.
\end{align*}

\subsection{Proof of Lemma~\ref{lem:all}}
\begin{reflemma}{lem:all}
Let $\sigma$ be a strategy in $\dmdp$ such that $\Pr{\sigma}{\dmdp}{s}(\thr(x))=1$.
Then \emph{all} runs realisable under $\sigma$ are in $\thr(x)$.
\end{reflemma}
\begin{proof}
Suppose the lemma does not hold, and let $\omega=s_0a_0s_1\ldots$ be a run realisable under $\sigma$
such that $\disc{\omega} = x - \varepsilon$ for some $\varepsilon>0$.
Let $M{:=}\sum_{i=0}^\infty \beta^{i}\max_{s,a} |\gain(s,a)|=\frac{\max_{s,a} |\gain(s,a)|}{1-\beta}$ and let $k$ be such that $\beta^{k+1} \cdot 2\cdot M < \varepsilon$. Let
$\omega'$ be any run of the form $s_0a_1\ldots a_ks_kb_{k+1}t_{k+1}b_{k+2}\dots$ (i.e. $\omega_k = \omega'_k$). Then, denoting $t_{k}{:=}s_k$ we have
\[
 \disc{\omega'} = \disc{\omega} - \big(\sum_{i=k+1}^\infty \beta^{i} \cdot\gain(s_{i-1},a_i)\big) + \big(\sum_{i=k+1}^\infty \beta^{i} \cdot\gain(t_{i-1},b_i)\big) \le x-\varepsilon +
   \beta^{k+1} \cdot 2\cdot M < x.
\]
However, the probability of such runs is
nonzero, a contradiction with the assumption that \mbox{$\Pr{\sigma}{\dmdp}{s}(\thr(x))=1$.}
\end{proof}

\subsection{Interreducibility of discounted MDPs and discounted games}
Let us first fix a discounted MDP
$\dmdp=(\states,\actions,\trans,\gain,\discount)$
We define a game $\game = (S, (S\times A), s_0, \trans_\game, R_\game, \sqrt{\beta})$
with%
\begin{itemize}
 \item $(s,(s,a))\in T_\game$ whenever $T(s,a)$ is defined;
 \item $((s,a),s')\in T_\game$ whenever $T(s,a)(s') > 0$; and 
 \item $R_\game(s) = 0$ for all $s\in S$, and $R_\game((s,a)) = \gain(s,a)$
\end{itemize}
For the clarity of presentation we first assume that $\sqrt{\beta}$ is a rational number of polynomial encoding size. Then we will show how to get rid of this assumption.

Let $\sigma$ be a strategy in $\dmdp$ such that $\Pr{\sigma}{\dmdp}{s}(\thr(x))=1$.
We define a strategy $\sigma^\game$ for the player 1 in $\game$ by $\sigma^\game(s_0(s_0,a_1)s_1\ldots s_n) = (s_n,a_{n+1})$
where $a_n\in A$ is an arbitrary action satisfying $\sigma(s_0a_1\ldots s_n)(a_{n+1}) > 0$.
For all player 2 strategies $\zeta_2$ we have that
to $\gamerun{\sigma^\game,\zeta_2}=s_0(s_0,a_1)s_1\ldots$ corresponds the run $\omega=s_0a_1s_1\ldots$ in $\dmdp$ which is realisable under $\sigma$, and
$\disc{\gamerun{\sigma^\game,\zeta_2}} = \disc{\omega}$.
Because every run $\omega$ realisable under $\sigma$ is in
$\thr(x)$, we have that $\disc{\gamerun{\sigma^\game,\zeta_2}} \ge x$.
For the other direction, let $\zeta_1$ be a winning player 1 strategy, by~\cite{ZP:games-graphs} we can assume
that it is is memoryless.
We define a strategy $\sigma$ for $\dmdp$ by $\sigma(s)=\zeta_1(s)$ for all $s\in S$.
Assume $\Pr{\sigma}{\dmdp}{s}(\thr(x)<1)$, and let $\omega=s_0a_1s_1\ldots$ be a run realisable under $\sigma$ such that
$\disc{\omega}<x$. Then we can fix a strategy $\zeta_2$ for player 2 in $\game$ defined by $\zeta_2(s_0(s_0,a_1)\ldots s_n) = a_{n+1}$.
We can easily show that $\disc{\gamerun{\zeta_1,\zeta_2}} = \disc{\omega} < x$, which contradicts that $\zeta_1$ is winning.

Now we drop the assumption that $\sqrt{\beta}$ is a rational number of polynomial encoding size. In such a case we represent the number $\sqrt{\beta}$ symbolically, as a triplet $(P(x),0,1)$, where $P(x)=x^2 - \beta$ is the minimal polynomial of $\sqrt{\beta}$ over $\Qset$ and the numbers $0,1$ represent the fact that we are interested in the positive (i.e., the one lying in the interval $[0,1]$) root of $P(x)$. Now consider again the aforementioned game $\game$ with $\sqrt{\beta}$ represented as above. Surely, determining the winner in such a game is at least as hard as determining the winner in ``standard'' discounted games, since all rational numbers can be also represented in the triplet form. It thus remains to show that the problem of determining the winner is in $\NP \cap \coNP$. Let us recall how the $\NP$-algorithm for standard games (see~\cite{ZP:games-graphs}) works: first, it guesses a winning memoryless deterministic strategy of player $1$ and then it verifies, using linear programming 
techniques, that against this strategy the player 2 cannot decrease the discounted reward below $x$. Now linear programs with coefficients represented in the triplet form can be solved on a Turing machine in time polynomial in the encoding size of the triplets and in the degree of the algebraic extension defined by adjoining all the coefficients in the program to $\Qset$ 
(see \citeadd[Theorem~21]{Beling}).
The linear program obtained by guessing the strategy in $\game$ contains only one coefficient which may be irrational, namely $\sqrt{\beta}$, which generates an extension of degree at most 2. Thus, we can again verify that the guessed strategy is winning in polynomial time. For the $\coNP$ upper bound we proceed similarly.

For the other direction of interreducibility, for a discount game $\game = (S_1, S_2, s_0, T_\game, R, \beta)$
we define a discounted MDP $\dmdp=(S_1,S_2,\trans,\gain,\beta^2)$
where
\begin{itemize}
 \item $\trans$ is an arbitrary function satisfying that $\trans(s,t)(s') > 0$ iff $(s,t)\in T$ and $(t,s')\in T$;
 \item $\gain(s,t) = R(s)/\beta + R(t)$
\end{itemize}

The rest of the proof proceeds in the same way as above.

\begin{remark}
If $\sqrt{\beta}$ is a rational number, then the pseudopolynomial algorithm for the qualitative problem in solvency MDPs can be immediately obtained from the pseudopolynomial algorithm for discounted games in \cite{ZP:games-graphs}. The algorithm in that paper iterates, for a pseudopolynomial number of steps, a suitable Bellman functional, where each iteration performs a polynomial number of additions, multiplications, divisions and comparisons, which involve the discount (i.e., in our reduction, the number $\sqrt{\beta})$. Since all of these operations can be computed in polynomial time for algebraic numbers in the triplet form \citeadd[Proposition 16]{Beling}, and all the intermediate results lie in the extension generated by $\sqrt{\beta}$ over $\Qset$, the algorithm is pseudopolynomial even for games with symbolically represented discounts. We note that in our game $\game$, the optimal value resulting from this algorithm is rational even if $\sqrt{\beta}$ is irrational, because it corresponds to the 
minimal threshold achievable with probability 1 in a discounted MDP with rational discount $\beta$. Rationality of this minimal threshold can be shown by devising a suitable Bellman functional (we omit this argument, because it is not essential for our paper).
\end{remark}

\section{Proofs for Section~\ref{sec:quantitative}}
\label{app:quantitative}

\subsection{Proof of Proposition~\ref{prop:lgug}}
\label{app:lgug}
Here we present an extended version of
Proposition~\ref{prop:lgug}.
\begin{proposition}\label{prop:lgug-ext}
For every state $s$ of the solvency MDP $\mdp$ there are rational numbers
 $U(\mdp,s)$ and $L(\mdp,s)$, such that
 \begin{align*}
  U(\mdp,s) &\mydef \arg\inf_{x\in \Rset} \forall \sigma \,.\,\mathbb{P}^{\sigma}_{\mdp,(s,x)}(\Win) = 1,\\
 L(\mdp,s) &\mydef \arg\sup_{x\in \Rset} \forall \sigma \,.\,\mathbb{P}^{\sigma}_{\mdp,(s,x)}(\Win) = 0,
 \end{align*}
 of encoding size polynomial in $\size{\mdp}$, and they are solutions of the following linear programs:
 \[
\begin{tabular}{lll}
   max                  & & $ \sum_{s\in\states} L(\mdp,s)$\\
   s.t.                 & &
   $ L(\mdp,s) \le \frac{1}{\interest}
  (L(\mdp,t) - \gain(s,a))$ \\
  & &(for all $s\in \states$, $a\in \actions(s)$, and $t\in\Succ(s,a)$)
\end{tabular}
\]
and
\[
\begin{tabular}{lll}
   min                  & & $ \sum_{s\in\states} U(\mdp,s)$\\
   s.t.                 & &
   $ U(\mdp,s) \ge \frac{1}{\interest}
  (U(\mdp,t) - \gain(s,a))$ \\
  & &(for all $s\in \states$, $a\in \actions(s)$, and $t\in\Succ(s,a)$).
\end{tabular}
\]

Moreover, we have $\mathbb{P}^{\sigma}_{\mdp,(s,U(M,s))}(\Win) = 1$ for every strategy $\sigma$.
\end{proposition}
\begin{proof}

First we show that these are actually  real numbers and not $\pm\infty$.
Let $g_{\max}=\max_{(s,a) \in \states\times \actions} \gain(s,a)$
be the maximal gain that occurs in the solvency MDP, and
fix arbitrary $x$ such that $g_{\max} + \interest\cdot x < 0$,
denoting $\tau := g_{\max} + \interest\cdot x$.
Then {\em any} run starting in $(s,x)$ is of the form
$(s_0,x_0)\cdot a_1\cdot (s_1, x_1)\cdot a_2\cdots$
where $x_i \le x + i\cdot \tau$. Hence we get that $L(\mdp,s) > x$.
For $U(\mdp,s)$ we take the minimal gain $g_{\min}=\min_{(s,a) \in \states\times \actions}
\gain(s,a)$ and proceed similarly.

We proceed by proving that the above values satisfy the
optimality conditions.
We first present the proof for value $U(\mdp,s)$.

Assume that the initial wealth is $\xinit \ge U(\mdp,s)$.
From LP it follows that for any action $a$
for all $t\in \Succ(s,a)$ we have
$\interest \cdot \xinit + \gain(s,a) \ge U(\mdp,t)$.
So no matter how the strategy picks actions, in the following states
the accumulated
wealth never falls below $\min_{s\in \states} U(\mdp,s)$, and because
we know that $U(\mdp,s)\in \Rset$, this ensures that any strategy
wins almost surely.

For the other direction assume that $\xinit < U(\mdp,s)$ and let
$\delta = U(\mdp,s) - \xinit$. We construct a strategy $\sigma$
which loses with positive probability. Let $\sigma$ pick an action
$a = \arg\max_{a\in\actions(s)}\max_{t\in\Succ(s,a)}
 \frac{1}{\interest}
  (U(\mdp,t) - \gain(s,a))$, and let
  $t = \arg\max_{t\in\Succ(s,a)}
 \frac{1}{\interest}
  (U(\mdp,t) - \gain(s,a))$, i.e., such that
  action $a$ and state $t$ is a bounding constraint in the LP.
It follows that
$U(\mdp,t) - (\interest \cdot \xinit + \gain(s,a)) = \interest \cdot \delta$.
  The strategy continues by picking actions the same way as in $s$
  and ensures that after $k$ steps, there exists a run ending in some state $t$
  which has a nonzero measure and the difference between $U(\mdp,t)$ and
  wealth  is equal to $\interest^k\cdot\delta$, and so
  for any value $X>-\infty$ we can find a $k$ such that the wealth $<X$ will
  be accumulated on some finite path having nonzero probability. This implies
  that wealth $< \min_{s\in\states}L(\mdp,s)$ will be eventually reached and thus the strategy will
  lose with positive probability.

Now we prove that the values $L(\mdp,s)$  satisfy the 
optimality conditions.
First, assume that the initial wealth $\xinit$ satisfies $\xinit < L(\mdp,\sinit)$
and let 
$\delta = L(\mdp,\sinit) - \xinit$. From the linear program
we know that for all actions $a$ and successors $t\in\Succ(\sinit,a)$ 
we have that 
$L(\mdp,t) - (\interest \cdot \xinit + \gain(\sinit,a)) 
\ge \interest \cdot \delta$. 
Hence, no matter what is the choice of the strategy, in the next 
step the difference between wealth and $L(\mdp,t)$ will be at least
$\interest \cdot \delta$ for any successor $t$. 
We can show by induction that after $k$ steps 
the difference between the wealth and $L(\mdp,t)$ is
at least $\interest^k\cdot \delta$; and because $\interest > 1$
we have that as $k\rightarrow \infty$ we have wealth
going to $-\infty$.

Now let the initial wealth be $\xinit > L(\mdp,\sinit)$ and let $\delta=\xinit-L(\mdp,\sinit)$.
We construct a strategy, which is winning with positive probability.
Consider the strategy $\sigma$ which picks an action
$a = \arg\min_{a\in\actions(\sinit)}\min_{t\in\Succ(\sinit,a)}
 \frac{1}{\interest}
  (L(\mdp,t) - \gain(\sinit,a))$, and let
  $t = \arg\min_{t\in\Succ(\sinit,a)}
 \frac{1}{\interest}
  (L(\mdp,t) - \gain(\sinit,a))$, i.e., such that
  action $a$ and state $t$ is a bounding constraint in the LP.
  Hence, it follows that
  $(\interest\cdot \xinit + \gain(\sinit,a)) - L(\mdp,t) = 
  \interest \cdot \delta$.
  The strategy continues by picking actions the same way as in $\sinit$
  and ensures that after $k$ steps, there exists a run ending in some state $t$ 
  which has a nonzero measure and the difference between wealth
  and the $L(\mdp,t)$ is equal to $\interest^k\cdot\delta$, and so
  for any value $X<\infty$ we can find a $k$ such that the wealth $>X$ will
  be accumulated on some finite path having nonzero probability. This implies
  that wealth $> \max_{s\in\states}U(\mdp,s)$ will be eventually reached and thus the strategy will
  win with positive probability.
  
  The bound on the encoding size of the numbers follows from the standard results on linear programming.
   \end{proof}
   
   \subsection{Supplement to Example~\ref{ex:infstate}}
\label{app:infstate}
  We distinguish two cases. Firstly, if $3k+2^{n}< 2^{n+1}$, we put $x_n=3k+2^{n}$.
  Secondly, if $3k+2^{n}\geq 2^{n+1}$, we have $3k-2^n = 3k + 2^n - 2^{n+1} \geq 0$.
  We argue that $3k - 2^n < 2^{n+1}$, which allows us to put $x_{n+1}=3k - 2^n$.
  Suppose the opposite, i.e. $3k - 2^n \geq 2^{n+1}$. This gives us
  $3k \geq 2^{n+1} + 2^n = 3\cdot 2^n$ and thus $k \geq 2^n$, again a contradiction.%

\subsection{Proof of Proposition~\ref{prop:finite-memory-reach}}
\label{app:bellman-gen}

Let us recall the fact that thanks to the Proposition~\ref{prop:lgug} we are now restricted to deterministic strategies.

\begin{refproposition}{prop:finite-memory-reach}
For every initial configuration $(s,x)$ and every $\eps>0$ there is a strategy $\sigma_{\eps}$ such that starting in $(s,x+\eps/2)$, $\sigma_{\eps}$ ensures hitting of a rentier configuration in at most $n=\big\lceil \frac{\log{(U(M)-L(M))}+\log{\eps^{-1}}+2}{\log \rho} \big\rceil$ steps with probability at least $\Val[M]{s,x}$. In particular, $\Prb^{\sigma_{\eps}}_{(s,x+\eps/2)}(\Win)\geq \Val[M]{s,x}$.
\end{refproposition}

In the proof of Proposition \ref{prop:finite-memory-reach} we proceed by a series of lemmas. First we show that the vector $\WRvec = (\WR{s,p})_{ p\in [0,1]}^{ s\in S} \in \Rset^{S\times[0,1]}$ is a unique fixed point of a suitable Bellman operator $\calL$. %

Let $s$ be any state of $M$. For an action $a\in A(s)$ and number $p \in [0,1]$ we denote by $B(s,a,p)$ the set of all vectors $\vec{q}\in [0,1]^{\Succ(s,a)}$ that satisfy $\sum_{s'\in \Succ(s,a)} \vec{q}{(s')}\cdot(T(s,a)(s'))\geq p$. The intuition behind the $B(s,a,p)$ vectors is that if a strategy $\sigma$ is $p$-winning in $(s,x)$ and it picks an action
$a$, then there must be a vector $\vec{q}\in B(s,a,p)$ such that for all $s'\in\Succ(s,a)$, the probability of winning from the successor of $(s,x)$ that is of the form $(s',x')$ for some $x'$ must be at least $\vec{q}(s')$.
Consider now the Bellman operator $\calL$ defined on the uncountably-dimensional space $\Rset^{S\times [0,1]}$ as follows:
\[
 \calL(\vec{V}){(s,p)} \quad = \quad \min_{a \in A(s)} \adjustlimits \inf_{\vec{q}\in B(s,a,p)} \max_{s' \in Succ(s,a)} \frac{1}{\rho}\cdot (\vec{V}{(s',\vec{q}{(s')})}-F(s,a)),
\]
for all vectors $\vec{V}\in\Rset^{S\times[0,1]}$ and all $(s,p)\in S\times [0,1]$.

\begin{lemma}
\label{lem:bellman-fixed-point}
 The vector $\WRvec$ is a fixed point of the operator $\calL$. 
\end{lemma}
\begin{proof}%
 Assume, for the sake of contradiction, that there are $s\in S$, $p\in[0,1]$ such that $\calL(\WRvec){(s,p)}< \WR{s,p}$. Pick an arbitrary $\delta>0$ such that $\calL(\WRvec)(s,p)+\delta<\WR{s,p}$, and denote by $x$ the left-hand side of this inequality. From the definition of $\calL$ it follows, that there are $a^* \in A(s)$ and $\vec{q}^* \in B(s,a^*,p)$ such that for all $s' \in \Succ(s,a^*)$ we have $\frac{1}{\rho}\cdot (\WR{s',\vec{q}^*{(s')}}-F(s,a^*)) \leq \calL(\WRvec){(s,p)} + \delta = x$, or in other words,
\begin{equation}\label{eq:one-step-optimal}
\rho \cdot x + F(s,a^*)> \WR{s',\vec{q}^*{(s')}}.
\end{equation}
Now, starting in $(s,x)$ the strategy can choose the action $a^*$ in the first step. If in the second step the current vertex is $s'$ (where $s'\in\Succ(s,a^*)$), we switch to a strategy that ensures winning from $(s',\rho\cdot x + F(s,a^*))$ with probability at least $\vec{q}^*{(s')}$ (such a strategy must exist, due to \eqref{eq:one-step-optimal}). Using this 
approach, the probability of winning from $(s,x)$ is at least $\sum_{
s'\in \Succ(s,a^*)} \vec{q}{(s')}\cdot(T(s,a^*)(s'))\geq p$, where the last inequality holds because $\vec{q}^*\in B(s,a^*,p)$. Since $x < \WR{s,p}$, we get a contradiction with the definition of $\WRvec$.

It remains to show that $\calL(\WRvec){(s,p)}\leq \WR{s,p}$, for an arbitrary fixed $(s,p)\in S\times [0,1]$. It suffices to show that for every $\eps>0$ there is a strategy $\sigma$ such that $\Prb^{\sigma}_{(s,\calL(\WRvec){(s,p)}+\eps)}\geq p$.
Similarly to the previous paragraph, there must be $a^*\in A(s)$ and $\vec{q}^*\in B(s,a^*,p)$ such that $\rho \cdot(\calL(\WRvec){(s,p)} +\eps)  + F(s,a^*)\geq \WR{s',\vec{q}^*{(s')}}$, for all $s'\in \Succ(s,a^*)$.
So if the strategy $\sigma$ chooses $a^*$ in the first step, then in the second step the play will be in some configuration $(s',\rho \cdot(\calL(\WRvec){(s,p)}+\eps)  + F(s,a^*))$, from which a strategy winning with probability at least $\vec{q}^*{(s')}$ exists, and $\sigma$ will behave as this strategy from the second step onwards. Since $\vec{q}^*\in B(s,a^*,p)$, it follows that indeed $\Prb^{\sigma}_{(s,\calL(\WRvec){(s,p)}+\eps)} (\Win)\geq p$.   %
\end{proof}
We denote by $\LU$ the set of all vectors $\vec{V}\in \Rset^{S\times [0,1]}$ that satisfy $L(M,s) \leq \vec{V}{(s,p)} \leq U(M,s)$, for all $(s,p)\in S\times[0,1]$. We also denote $\norm{\vec{V}}=\sup_{(s,p)\in \states\times[0,1]}|\vec{V}(i)|$.

\begin{lemma}
\label{lem:contraction}
 If $\vec{V}\in \LU$, then also $\calL(\vec{V})\in \LU$.
 Moreover, for every pair of vectors $\vec{V},\vec{V}'$ we have $\norm{\calL(\vec{V})-\calL(\vec{V}')} \leq \frac{1}{\rho}\norm{\vec{V}-\vec{V}'}$.
\end{lemma}
\begin{proof}
Let $\vec{V}\in\LU$, $s\in S$ and $p\in[0,1]$ be arbitrary. Assume, for the sake of contradiction, that $\calL(\vec{V}){(s,p)} > U(M,s)$. By definition, any strategy wins from $(s,\calL(\vec{V}){(s,p)})$ with probability $1$. Thus, for every $a\in A(s)$ end every $s'\in \Succ(s,a)$ we have $\rho\calL(\vec{V}){(s,p)} + F(s,a) \geq U(M,s')$. But at the same time, by definition of $\calL$ we have $\calL(\vec{V}){(s,p)} \leq \rho^{-1}(\vec{V}{(s',p')}-F(s,a))$ for suitable $p' \in [0,1]$. Combining these two inequalities we get $\vec{V}{(s',p')}\geq U(M,s')$, a contradiction with $\vec{V}\in \LU$. The inequality $\calL(\vec{V}){(s,p)} \geq L(M,s)$ can be established in a similar way.

For the second part, fix arbitrary vectors $\vec{V}$, $\vec{V}'\in \Rset^{S\times [0,1]}$ and some $(s,p)\in S\times [0,1]$. We have to show that $|\vec{V}(s,p)-\vec{V}'{(s,p)}|\leq \rho^{-1}\norm{\vec{V}-\vec{V}'}$. Let us choose an arbitrary $\eps>0$. From the definition of $\calL$ it follows that there are $a,b \in A(s)$, $\vec{q} \in B(s,a,p)$ and $\vec{r} \in B(s,b,p)$ such that
\begin{align}
\label{eq:contraction1}
 y_1 &\mydef \max_{s' \in \Succ(s,a)} \frac{1}{\rho}\big(\vec{V}{(s',\vec{q}{(s')})} - F(s,a)\big) \leq \vec{V}{(s,p)} + \frac{\eps}{2}\\
\label{eq:contraction2}
 y_2 &\mydef \max_{s' \in \Succ(s,b)} \frac{1}{\rho}\big(\vec{V}'{(s',\vec{r}{(s')})} - F(s,b)\big) \leq \vec{V}'{(s,p)} + \frac{\eps}{2}.
\end{align}
Assume that $y_1 \geq y_2$, the other case can be handled in a symmetric way. We have
\begin{align}
 0 &\leq y_1 - y_2 \leq \max_{s' \in \Succ(s,b)} \frac{1}{\rho}\big(\vec{V}{(s',\vec{r}{(s')})} - F(s,b)\big) + \frac{\eps}{2} - y_2 \notag\\
 &= \frac{1}{\rho} \max_{s'\in \Succ(s,b)}(\vec{V}{(s',\vec{r}{(s')})} - \vec{V}'{(s',\vec{r}{(s')})}) + \frac{\eps}{2} \leq \frac{1}{\rho}\norm{\vec{V}-\vec{V}'} + \frac{\eps}{2},\label{eq:contraction3}
\end{align}
where the second inequality follows from the facts that $\vec{V}{(s,p)} \leq y_1 \leq \vec{V}{(s,p)} + \eps/2$ and $\vec{V}{(s,p)}\leq \max_{s' \in \Succ(s,b)} \frac{1}{\rho}\big(\vec{V}{(s',\vec{r}{(s')})} - F(s,b)\big)$. Putting \eqref{eq:contraction1}, \eqref{eq:contraction2} and \eqref{eq:contraction3} together, and using the fact that $y_2 \geq \vec{V}'{(s,p)}$, we obtain
\begin{align*}
|\vec{V}{(s,p)}-\vec{V}'{(s,p)}| \leq |y_1 - y_2| + \frac{\eps}{2} \leq \frac{1}{\rho}\norm{\vec{V}-\vec{V}'} + \eps.
\end{align*}
Since $\eps>0$ was chosen arbitrarily, the result follows.
\end{proof}

The previous lemma shows, that the operator $\calL$ is a contraction mapping on $\LU$. Now the set $\LU$ equipped with the supremum norm $\norm{\cdot}$ is a Banach space. By the Banach fixed-point theorem, $\calL$ has a unique fixed point in $\LU$, which must be equal to $\WRvec$ by Lemma \ref{lem:bellman-fixed-point}. Moreover, for every vector $\vec{V}\in \LU$ the sequence $\calL^{n}(\vec{V})$ converges to this fixed point as $n$ approaches infinity.

Consider now a vector $\vec{V}^0 \in \LU$ such that $\vec{V}^0{(s,p)} = U(M,s)$ for every $(s,p)$.

\begin{lemma}
\label{lem:n-step-values}
Consider any $(s,p)\in S\times[0,1]$ and any $y>\calL^{n}(\vec{V}^0){(s,p)}$. Then there is a strategy $\sigma$ such that starting in $(s,y)$, the strategy $\sigma$ ensures hitting a rentier configuration in at most $n$ steps with probability at least $p$.
\end{lemma}
\begin{proof}
We proceed by induction on $n$. The case $n=0$ is trivial. So assume that and that the lemma holds for some $n\in \Nset_0$. Consider any $(s,p)\in S\times[0,1]$. Then there must be an action $a\in A(s)$ and vector $\vec{r} \in B(s,a,p)$ such that 
\begin{equation}
\label{eq:nstep1}
y > \max_{s'\in \Succ(s,a)} \frac{1}{\rho} \big(\calL^{n}(\vec{V}^0){(s',\vec{r}{(s')})} - F(s,a) \big).
\end{equation}
Then, in order to reach a rentier configuration in at most $n+1$ steps with probability at least $p$, the strategy can proceed as follows: in the first step, it chooses the aforementioned action $a$. In the second step, the play is in some state $s'$ with probability $T(s,a)(s')$ and the current wealth is greater than $\calL^{n}(\vec{V}^0){(s',\vec{r}{(s')})}$ (by \eqref{eq:nstep1}). By induction, the strategy can then switch to a strategy that reaches a rentier configuration in at most $n$ steps with probability at least $\vec{r}{(s')}$. Because $\vec{r}\in B(s,a,p)$, it follows that this strategy ensures reaching a rentier configuration from $(s,y)$ in at most $n+1$ steps with probability at least $p$.
\end{proof}

We can now finish the proof of Proposition \ref{prop:finite-memory-reach}. 
We have $\norm{\WRvec-\calL^n(\vec{V}^0)} \leq \frac{1}{\rho^n}\cdot \norm{\WRvec-\vec{V}^0}\leq \frac{1}{\rho^n}\cdot(U(M)-L(M))$ (where the first inequality follows from Lemmas \ref{lem:bellman-fixed-point} and \ref{lem:contraction}). It follows that  for $n=
\big\lceil \frac{\log{(U(M)-L(M))}+\log{\eps^{-1}}+2}{\log \rho} \big\rceil$ it holds $\norm{\WRvec-\calL^{n}(\vec{V}^0)} \leq \eps/4$. In particular, $x+\eps/2 \geq \WR{s,\Val[M]{s,x}}+\eps/2> \calL^{n}(\vec{V}^0){(s,\Val[M]{s,x})}$.
Thus, the strategy $\sigma_{\eps}$ can be chosen to be the strategy $\sigma$ from Lemma~\ref{lem:n-step-values} for $p=\Val{s,x}$, $n$ and $y = x+\eps/2$.

\subsection{Proof of Lemma~\ref{lem:discretization}}
\label{sec:discretization}

\begin{reflemma}{lem:discretization}
Let $(s,y)$ be an arbitrary configuration of $M$. Then the following holds:
\begin{enumerate}
 \item For every $\sigma\in \St[M]$ there is $\pi\in \St[M_{\lambda,n}]$ such that
  $$\Prb^{\pi}_{M_{\lambda,n},([s,y]_{\lambda},0)}(\Hit) \geq \Prb^{\sigma}_{M,(s,y)}(\Cover{0}).$$
 \item There is $\sigma\in \St[M]$ such that
 $$\Prb^{\sigma}_{M,(s,y)} (\Cover{n\cdot\lambda\cdot\rho^{n}}) \geq \sup_{\pi}\,\Prb^{\pi}_{M_{\lambda,n},([s,y]_{\lambda},0)} (\Hit) \mydef v,$$
  where the supremum is taken over $\St[M_{\lambda,n}]$. Moreover, the number $v$ and a finite representation of the strategy $\sigma$ can be computed in time $\size{M_{\lambda,n}}^{\calO(1)}$.
\end{enumerate}
\end{reflemma}
\begin{proof}
Again, we remind the reader that we are now restricted to deterministic strategies (because of Proposition~\ref{prop:lgug}).
For the purpose of this proof, let us define the \emph{absorbing vertices} of $M_{\lambda,n}$ to be all vertices of this MDP that are not of the form $(s,D,i)$ with $D$ bounded interval and $i < n-1$. Moreover, we denote by $B$ the set of all histories in which the last vertex is of the form $(t,(U(M,t),\infty),i)$ while all the previous vertices are non-absorbing. We also denote by $C$ the set of all histories in $M$ that contain exactly one rentier configuration (and thus, their every proper prefix does not contain any such configuration). Note that $\Hit = \bigcup_{u \in B} \Cone(u)$ and $\Cover{0}=\bigcup_{v \in C} \Cone(v)$. In the following we assume that the initial configuration $(t_0,x_0)$ of all runs in $M$ is equal to $(s,y)$ and that the initial vertex $(s_0,D_0,0)$ of all runs in $M_{\lambda,n}$ is equal to $([s,y]_{\lambda},0)$.

First we describe certain natural correspondence between runs in $M$ and $M_{\lambda,n}$, which will be used throughout the proof. Let $X$ be a set of all histories in $M_{\lambda,n}$ that do contain at most one absorbing vertex. For any history $X \ni u = (s_0,D_0,0)a_1 \cdots a_{k}(s_k,D_k,k)$ in $M_{\lambda,n}$ there is exactly one history $f(u)=(t_0,x_0)b_1\cdots b_{k}(t_k,x_k)$ in $M$ such that, $s_i = t_i$ for all $0\leq i \leq k$ and $a_i=b_i$ for all $1\leq i \leq k$. Note that $f$, viewed as a function from $X$ to the histories of $M$, is injective.%

On the other hand, for every history $v=(t_0,x_0)b_1\cdots b_{k}(t_k,x_k)$ in $M$  there is a unique history $g(v)=(s_0,D_0,0)a_1 \cdots a_{k}(s_k,D_k,k)$ in $M_{\lambda,n}$ such that for all $0 \leq i \leq k$ the vertex $(s_i,D_i,i)$ is either absorbing, or $s_i = t_i$ and $a_{i+1}=b_{i+1}$.\footnote{The equality of actions is considered only if $i\neq k$} A straightforward induction reveals, that for all $0 \leq i \leq k$ we have $w_{D_i}\geq x_i$ (we always round the numbers \emph{up} to the nearest multiple of $\lambda$). It follows that $\Hit \supseteq \bigcup_{v \in C} \Cone(g(v))$. The function $g$ viewed as a mapping from $C$ to histories in  ${M_{\lambda},n}$, may not be injective. Below, we use $\kernel\, g$ to denote the kernel of $g$, and $[v]_g$ to denote the equivalence class of $v$ in $C/\kernel\, g $. %

(1.) Fix some strategy $\sigma$ in $M$. We define strategy $\pi$ as follows: for every history $u$ that does not contain an absorbing vertex we set $\pi(u)=\sigma(f(u))$. For other histories (i.e. those that reach an absorbing vertex from which there is no escape), $\pi$ can choose any action. It is easy to verify that for every history $v$ in $M$ we have $$\sum_{u\in[v]_g}\Prb^{\sigma}_{M,(s,y)} \big(\Cone(u)\big)= \Prb^{\pi}_{M_{\lambda,n},[s,y]_{\lambda}} \big(\Cone(g(v))\big).$$ It follows that
\begin{align*}
 \Prb^{\pi}_{[s,y]_{\lambda}}(\Hit) &\geq \Prb^{\pi}_{[s,y]_{\lambda}}\left( \bigcup_{v \in C} \Cone(g(v))\right)=  \Prb^{\pi}_{[s,y]_{\lambda}}\left(\bigcup_{[v]_g \in C/\kernel g } \Cone(g(v)) \right)\\&= \sum_{[v]_g \in C/\kernel g} \Prb^{\pi}_{[s,y]_{\lambda}} \big(\Cone(g(v))\big) =\sum_{v \in C}\sum_{u\in [v]_g} \Prb^{\sigma}_{(s,y)} \big(\Cone(u)\big)\\ &= \sum_{v \in C}\Prb^{\sigma}_{(s,y)} \big(\Cone(v)\big)= \Prb^{\sigma}_{(s,y)}(\Cover{0}),
\end{align*}
where the first equality on the second line follows from the fact that $\Cone(h)$ and $\Cone(h')$ are disjoint events for $h\neq h'$ when $h$ is not a prefix of $h'$ and vice versa.

(2.) By standard results on MDPs with reachability objectives there is a memoryless strategy $\pi^*$ such that $\Prb^{\pi^*}_{[s,y]_{\lambda}} (\Hit)= \sup_{\pi}\Prb^{\pi}_{[s,y]_{\lambda}} (\Hit)$, where the maximum is taken over all strategies. It thus suffices to prove that there is a strategy $\sigma$ in $M$ satisfying $\Prb^{\sigma}_{(s,y)} (\Cover{n\cdot\lambda\cdot\rho^{n+1}}) \geq \Prb^{\pi^*}_{[s,y]_{\lambda}} (\Hit)$.

Let $u = (s_0,D_0,0)a_1 \cdots a_{k}(s_k,D_k,k)$, $u\in B$ be a history in $M_{\lambda,n}$ and let $f(u) = (t_0,x_0)b_1\cdots b_{k}(t_k,x_k)$ be the corresponding history in $M$. We prove by induction on $i$ that for every $0 \leq i < k$ we have $w_{D_k}\leq x_k + (i+1)\lambda\rho^i$. The case $i=0$ is trivial, since $(s_0,D_0)=[t,y]_{\lambda}$ and $(t_0,x_0)=(t,y)$, so $w_{D_0}-x_0 \leq \lambda$. Suppose now that $i>0$ and that $w_{D_{i-1}} - x_{i-1} \leq i\lambda\rho^{i-1}$.
\begin{align}
 w_{D_i} - x_i &\leq \rho w_{D_{i-1}} + F(s_{i-1},a_{i-1}) + \lambda - x_i \notag\\ &= \rho w_{D_{i-1}} + F(s_{i-1},a_{i-1}) + \lambda - \rho x_{i-1} - F(t_{i-1},b_{i-1})
 \notag\\ &=\rho(w_{D_{i-1}}-x_{i-1}) + \lambda \leq i\lambda\rho^i + \lambda \leq (i+1)\lambda\rho^i.\label{eq:discretization1}
\end{align}
Here, the first inequality follows from the fact that $D_i$ is an interval of length $\lambda$ containing $\rho w_{D_{i-1}} + F(s_{i-1},a_{i-1})$, the first equality on the third line follows from $t_{i-1}=s_{i-1}$ and $b_{i-1}=a_{i-1}$, while the next inequality follows from the induction hypothesis. As a consequence, we have
\begin{align*}
 x_k &= \rho x_{k-1} + F(t_{k-1},b_k) = \rho x_{k-1} + F(s_{k-1},a_{k-1}) \\&\geq \rho w_{D_{k-1}}  + F(s_{k-1},a_{k-1})- k\lambda\rho^{k+1} \geq U(M,s_k) - k\lambda\rho^{k},
\end{align*}
where the first inequality on the second line follows from \eqref{eq:discretization1}, while the next inequality follows from the fact that the last vertex of history $u$ is of the form $(t,(U(M,t),\infty),k)$.

Thus, $\Cover{n\lambda\rho^{n}}\supseteq \bigcup_{u \in B}\Cone(f(u))$. We now proceed similarly as in (1.). We define strategy $\sigma$ as follows: for a given history $v$, if $g(v)\in X$, $\sigma(v)=\pi^*(g(v))$, while for the other histories, $\sigma$ chooses an arbitrary action. Note that this representation of $\sigma$ can be computed in time polynomial in $\size{M_{\lambda,n}}$, since we can use standard polynomial-time algorithm for MDPs with reachability objectives to compute $\pi^*$, and $g(v)$ can be computed in time linear in length of $v$.\footnote{It can be actually computed online as new configurations are visited during the play.} It is easy to verify, that for every history $u \in X$ we have $\Prb^{\pi^*}_{M_{\lambda,n},[s,y]_{\lambda}} (\Runs(u))= \Prb^{\sigma}_{M,(s,y)}\big(\Runs(f(u))\big)$. Combining the previous observations we get
\begin{align*}
 \Prb^{\sigma}_{(s,y)} \big(\Cover{n\lambda\rho^{n}}\big) &\geq \Prb^{\sigma}_{(s,y)}\left( \bigcup_{u \in B} \Cone(f(u)) \right)= \sum_{u \in B} \Prb^{\sigma}_{(s,y)}\big(\Cone(f(u)) \big) \\
 &= \sum_{u \in B}\Prb^{\pi^*}_{[s,y]_{\lambda}} (\Cone(u)) = \Prb^{\pi^*}_{[s,y]_{\lambda}} (\Hit),
\end{align*}
where in the last equality on the second line we use the fact that $f$ is injective.
\end{proof}

\subsection{Proof that $n\cdot \lambda \cdot \rho^{n} \le \frac{\eps}{2}$}
\label{app:neps}
We have
\begin{align*}
 n\cdot \lambda \cdot \rho^{n}
  &\leq \frac{\eps^3}{64(U(M)-L(M))^2}\cdot\rho^{ \frac{\log{(U(M)-L(M))}+\log{\eps^{-1}}+2}{\log \rho} +1}\\
  &\leq \frac{\eps^3}{64(U(M)-L(M))^2}\cdot 2^{\log{(U(M)-L(M))}+\log{\eps^{-1}}+2}\cdot\varrho\\
  &\leq {\frac{\rho\cdot\eps}{4(U(M)-L(M))}} \cdot\frac{\eps}{2}\\
  &\leq \frac{\eps}{2}
\end{align*}
 where the last inequality holds because we assumed that $\log\rho< \log(U(M)-L(M))+\log(\eps^{-1})+2$.

\subsection{Complexity Analysis for Theorem~\ref{thm:quantitative-algorithm}}
\label{sec:alg-complexity}
Here we conclude the proof of Theorem~\ref{thm:quantitative-algorithm}.
From the previous observations we have that the complexity is $\size{M_{\lambda,n}}^{\calO(1)}$, which can be rewritten as
\begin{align*}
\left( |S|\cdot|A| \cdot\log\left(p_{\min}^{-1}\right)\cdot n\cdot\lambda^{-1} \right)^{\calO(1)} 
= \left( |S|\cdot|A|\cdot\log\left(p_{\min}^{-1}\right) \cdot \frac{U(M)-L(M)}{\eps\cdot\log\rho} \right)^{\calO(1)}.
\end{align*}
Noting that $U(M)-L(M)\leq 2 r_{\max}/(\rho-1)$ and $1/\log\rho \leq (1+1/(\rho-1))$, we conclude that the complexity is indeed polynomial in $|S|\cdot|A| \cdot\log\left(p_{\min}^{-1}\right)$ and exponential in $\log(|r_{\max}|/(\rho-1))$ and $\log(1/\eps)$.

\subsection{Proof of Theorem~\ref{thm:quant-hardness}}
\label{app:hardness}

\begin{reftheorem}{thm:quant-hardness}
The problem of deciding whether $\WR{s,p} \leq x$ for a given $x$ is $\NP$-hard. Furthermore, existence of any of the following algorithms is not possible unless P$=$NP:
 \begin{enumerate}
  \item An algorithm approximating $\WR{s,p}$ up to the absolute error $\delta$ in time polynomial in $|S|\cdot|A|\cdot\log\left(p_{\min}^{-1}\right)$ and $\log(|r_{\max}|/(\rho-1))$ and exponential in $\log(1/\delta)$.
  \item An algorithm approximating $\WR{s,p}$ up to the absolute error $\delta$ in time polynomial in $|S|\cdot|A|\cdot\log\left(p_{\min}^{-1}\right)$ and $\log(1/\delta)$ and exponential in $\log(|r_{\max}|/(\rho-1))$.
 \end{enumerate}
Above, the numbers $r_{\max}$ and $p_{\min}$ are as in Theorem~\ref{thm:quantitative-algorithm}.
\end{reftheorem}
\begin{proof}
We begin with part (1.), the second part is very similar.
We give the proof by reduction from the Knapsack problem. Let us have an instance of a Knapsack problem with items $1,\ldots, n$ (we assume $n\geq 2$), where the weight and value of the $i$th item is $w_i$ and $v_i$, respectively, and where the bound on the weight and value of the items to be put in the knapsack are $W$ and $V$, respectively. We denote $w_{\tot}=\sum_{1 \leq i \leq n} w_i$ and $v_{\tot}=\sum_{1 \leq i \leq n} v_i$. Without loss of generality we assume that: all the numbers $v_i$, $w_i$ are nonzero, and that the item weights are integers (this restriction of Knapsack is still NP-hard); and that $v_{\tot} < 1/n^2$ (otherwise we can transform the instance by dividing all the numbers $v_i$ and number $V$ by $v_{\tot}\cdot n^2$, without influencing the existence of a solution). %

We show how to compute, in time polynomial in the encoding size of the Knapsack instance, a solvency MDP \mbox{$M=(S, A,T, F, \rho)$} with an interest rate $\rho=1+\frac{1}{4n^2}$, and a number $p$ such that there is a solution to the instance of Knapsack if and only if $\WR{s_1,p}=0$ (for some distinguished state $s_1$ of $M$). The interest rate $\rho$ is chosen in such a way that the inequality $\frac{\rho^{2n}}{4} \leq \frac{1}{2}$ holds.

First, we put $S =\{s_i,s^+_i,s^-_i\mid 1 \leq i \leq n\}\cup\{s_{n+1},t_1,t_2,t_3\}$ and $A = \{a^+_i,a^-_i\mid 1\leq i \leq n\}\cup\{b\}$. %

Next, we set $\alpha=1/n^2$ and define the transitions as follows:
\begin{itemize}
 \item $T(s_i, a^+_i) = [s^+_i\mapsto 1]$, for all $1 \leq i \leq n$;
 \item $T(s_i, a^-_i) = [s^-_i\mapsto 1]$, for all $1 \leq i \leq n$;
 \item $T(s^+_i, b) = [t_1 \mapsto \frac{v_i}{1-(i-1)\alpha},\,t_2 \mapsto \frac{\alpha-v_i}{1-(i-1)\alpha},\,s_{i+1} \mapsto 1 - \frac{\alpha}{1-(i-1)\alpha}]$ for all $1\leq i\le n$ (note that this indeed defines a probability distribution, since $v_{\tot} \leq \alpha < 1- n\alpha$);
 \item $T(s^-_i, b) = [t_3 \mapsto \frac{\alpha}{1-(i-1)\alpha},\,s_{i+1}  \mapsto 1-\frac{\alpha}{1-(i-1)\alpha}]$ for all $1\le i\le n$;
 \item $T(t_1,b)=[t_1 \mapsto 1]$, $M(t_2,b)=[t_2 \mapsto 1]$, and $M(t_3,b)=[t_3 \mapsto 1]$;  %
 \item $T(s_{n+1},b)=[s_{n+1}\mapsto 1]$. 
\end{itemize}

The rewards are defined in the following way:
\begin{itemize}
 \item $F(s^-_{i},b)=w_i\cdot\rho^{-2(n-i)}$, for all $1\leq i \leq n$;
 \item $F(t_1,b)=1$, $F(t_2,b)=F(t_3,b)=-2(w_{\tot}+1)$;
 \item $F(s_{n+1},b)=-(w_{\tot}-W)/4n^2$. This ensures that $U(M,s_{n+1})=w_{\tot} - W$, so a run that visits $s_{n+1}$ is winning if and only if the current wealth upon entering $s_{n+1}$ is at least $w_{\tot} - W$.
 \item All other rewards are zero.
\end{itemize}

Figure \ref{fig:hardnessp} illustrates the construction on a simple example.
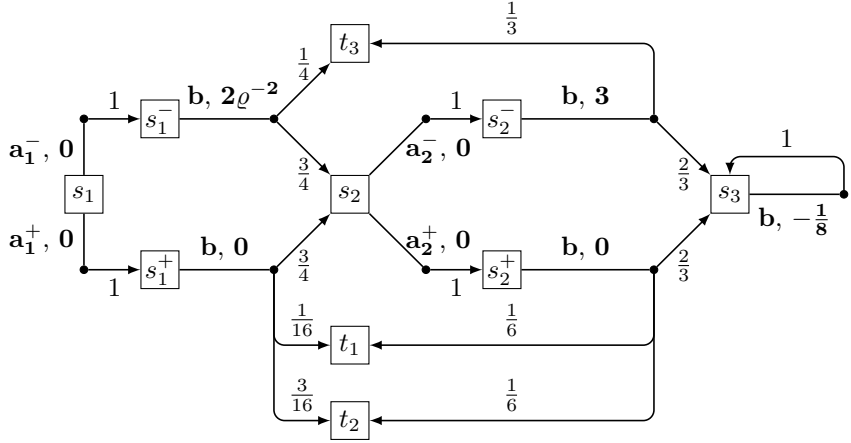
\begin{figure}
\centering

\begin{tikzpicture}
\node[good state] at (0,0) (s1) {$s_1$};
\node[good state] at (3.5,0) (s2) {$s_2$};
\node[good state] at (8.5,0) (s3) {$s_3$};

\node[good state] at (1,1) (sm1) {$s_1^-$};
\node[good state] at (5.5,1) (sm2) {$s_2^-$};
\node[good state] at (1,-1) (sp1) {$s_1^+$};
\node[good state] at (5.5,-1) (sp2) {$s_2^+$};
\node[good state] at (3.5,2) (t3) {$t_3$};
\node[good state] at (3.5,-2) (t1) {$t_1$};
\node[good state] at (3.5,-3) (t2) {$t_2$};

\node[distr] at (0,1) (d1-) {};
\node[distr] at (0,-1) (d1+) {};
\node[distr] at (2.5,1) (d1--) {};
\node[distr] at (2.5,-1) (d1++) {};
\node[distr] at (4.5,1) (d2-) {};
\node[distr] at (4.5,-1) (d2+) {};
\node[distr] at (7.5,1) (d2--) {};
\node[distr] at (7.5,-1) (d2++) {};
\node[distr] at (10,0) (df) {};

\draw[trarr] (s1) -- (d1-) node[midway, left] {$\mathbf{a^-_1}$, $\mathbf{0}$} -- (sm1) node[midway, above] {$1$};
\draw[trarr] (s2) -- (d2-) node[midway, right] {$\mathbf{a^-_2}$, $\mathbf{0}$} -- (sm2) node[midway, above] {$1$};
\draw[trarr] (s1) -- (d1+) node[midway, left] {$\mathbf{a^+_1}$, $\mathbf{0}$} -- (sp1) node[midway, below] {$1$};
\draw[trarr] (s2) -- (d2+) node[midway, right] {$\mathbf{a^+_2}$, $\mathbf{0}$} -- (sp2) node[midway, below] {$1$};
\draw[trarr] (sp1) -- (d1++) node[midway, above] {$\mathbf{b}$, $\mathbf{0}$} -- (s2) node[midway, below] {$\frac{3}{4}$};
\draw[trarr] (sp2) -- (d2++) node[midway, above] {$\mathbf{b}$, $\mathbf{0}$} -- (s3) node[midway, below] {$\frac{2}{3}$};
\draw[trarr] (sm1) -- (d1--) node[pos=0.6, above] {$\mathbf{b}$, $\mathbf{2\rho^{-2}}$} -- (s2) node[midway, below] {$\frac{3}{4}$};
\draw[trarr] (sm2) -- (d2--) node[midway, above] {$\mathbf{b}$, $\mathbf{3}$} -- (s3) node[midway, below] {$\frac{2}{3}$};
\draw[trarr] (d1--) -- (t3) node[midway, above] {$\frac{1}{4}$};
\draw[trarr] (d2--) -- (t3-|d2--) -- (t3) node[midway, above] {$\frac{1}{3}$};
\draw[trarr] (d1++) -- (t1-|d1++) -- (t1) node[midway, above] {$\frac{1}{16}$};
\draw[trarr] (d2++) -- (t1-|d2++) -- (t1) node[midway, above] {$\frac{1}{6}$};
\draw[trarr] (d1++) -- (t2-|d1++) -- (t2) node[midway, above] {$\frac{3}{16}$};
\draw[trarr] (d2++) -- (t2-|d2++) -- (t2) node[midway, above] {$\frac{1}{6}$};
\draw[trarr] (s3) -- (df) node[midway, below] {$\mathbf{b}$, $\mathbf{-\frac{1}{8}}$} -- (10,0.5) -- (8.5,0.5) node[midway, above] {$1$} -- (s3);
\end{tikzpicture}
\caption{Solvency MDP $M$ constructed for a simple Knapsack instance with items $1,2$ such that $w_1 = 2$, $w_2 = 3$, $v_1 = 1/16$, $v_2=1/8$, $W=3$ and $V=1/8$. We have $\rho = 1+1/16$ and $p=7/8$. To achieve greater compactness the picture omits loops under action $b$ on states $t_1$, $t_2$ and $t_3$. Action $b$ is rewarded with $-12$ in states $t_2$ and $t_3$.}
\label{fig:hardnessp}
\end{figure}

Finally, we set $p=1+V-1/n$. Clearly, the MDP $M$ and number $p$ can be computed in polynomial time. Note that $p>1$ would imply $V>v_{\tot}$, so in this case the Knapsack does not admit a solution. Thus, we assume that $p\in [0,1]$. %

Under any strategy the state $s_{i+1}$ is reached with probability $(1-i\alpha)$. This can be shown by induction on $i$,
since
\begin{eqnarray*}
 (1-i\alpha)(1-\frac{\alpha}{1-i\alpha})
 &=& 1 - \frac{\alpha}{1-i\alpha} - i\alpha + \frac{i\alpha^2}{1-i\alpha}
 = 1- \frac{\alpha + i\alpha - i^2\alpha^2 - i\alpha^2}{1-i\alpha}\\
 &=& 1-\frac{(i+1)\alpha(1-i\alpha)}{1-i\alpha}
 = 1-(i+1)\alpha
\end{eqnarray*}
Hence we get that the state $s_{n+1}$ is reached with probability $(1-1/n)$.
Further, the probability that
a run contains both $s_i$ and $t_1$ is $v_i$
if the strategy picks $a^+_i$ in the state $s_{i}$, and $0$ otherwise. So let us choose any (deterministic
 ) strategy $\sigma$ and an arbitrary initial configuration $(s_1,x)$ with $x\in [0,1/4]$. Obviously, we interpret the choice of action $a^+_i$ by $\sigma$ in $s_i$ as choosing the item $i$ to be packed into the knapsack. Denote $I^{+}_{\sigma}\subseteq \{1,\dots,n\}$ the set of all indexes $i$ such that $\sigma$ chooses action $a^{+}_i$ in $s_i$ (i.e. the set of items chosen to be packed into the knapsack). A straightforward induction on $i$ reveals that conditional on reaching the state $s_i$, the current wealth upon reaching this state is
$$ \sum_{\substack{j \in \{1,\dots,n\}\setminus I^+_{\sigma}\\ j<i }} w_j\cdot \rho^{-2(n-i+1)} +  \rho^{2(i-1)}\cdot x.$$
This has two consequences. First, until the run reaches one of the states $t_1$, $t_2$ or $s_{n+1}$ the current wealth is always bounded by $0$ from below and by $w_{\tot}+1$ from above. Thus, a run starting in $(s_1,x)$ is winning if it reaches $t_1$ and losing if it reaches $t_2$ or $t_3$. Second, if the run reaches the state $s_{n+1}$ (this happens with probability $(1-1/n)$), then the current wealth upon reaching this state is equal to 
\begin{equation}\label{eq:hardness1}\sum_{i \in \{1,\dots,n\}\setminus I^{+}_{\sigma}} w_i + \rho^{2n}\cdot x.\end{equation}

Thus, under $\sigma$ the conditional probability of winning on condition that the play reaches $s_{n+1}$ is 1 if and only if the items in $\{1,\dots,n\}\setminus I^+_{\sigma}$ have total weight at least $w_{\tot} - W - \rho^{2n}\cdot x$, or in other words, if the items in $I^{+}_{\sigma}$have total weight at most $W+\rho^{2n}\cdot x \leq W+1/2$. Otherwise, this conditional probability is 0. Since we assume that the item weights are integers, the aforementioned probability is 1 iff the items contained in $I^{+}_{\sigma}$ have total weight at most $W$. 

We claim that for any strategy $\sigma$ the set $I^+_{\sigma}$ is a solution of the original instance of Knapsack (i.e. set of items selected for inclusion into the knapsack, which satisfies the usual constraints on weight and value) if and only if $\Prb^{\sigma}_{(s_1,x)}(\Win)\geq p$.

First suppose that $I^+_{\sigma}$ is a solution to the Knapsack instance. Since the total weight of items in $I^+_{\sigma}$ is at most $W$, from \eqref{eq:hardness1} it follows that under $\sigma$, once the play reaches $s_{n+1}$ (this happens with probability $1-1/n$) the current wealth is at least $w_{\tot} - W = U(M,s_{n+1})$, so conditional on reaching $s_{n+1}$ the probability of winning is 1. Moreover, we know that every run that reaches $t_1$ is winning as well. We know that this happens with probability $\sum_{i \in I^{+}_{\sigma}} v_i \geq V$. We conclude that by using strategy $\sigma$ we win from $(s_1,x)$ with probability at least $1-1/n+V = p$.

On the other hand, if $I^{+}_{\sigma}$ is not a solution of the original instance, there are two cases to consider. Either  $\sum_{i \in \{1,\dots,n\}\setminus I^{+}_{\sigma}} w_i < w_{\tot} - W$, in which case the probability of winning under $\sigma$ is at most $v_{\tot}$, because the conditional probability of winning upon reaching $s_{n+1}$ is $0$, and the only other way to win is to reach $t_1$. Or $\sum_{i \in I^{+}_{\sigma}} v_i < V$, in which case the probability of winning is strictly smaller than $1-1/n+V$ (probability of reaching $s_{n+1}$, where the conditional probability of winning can be 1, plus the probability of visiting $t_1$ which is $\sum_{i \in I^{+}_{\sigma}} v_i < V$). In either case, the probability of winning is less than $p$.

It follows that the instance of Knapsack admits a solution if and only if there is $\sigma$ such that $\Prb^{\sigma}_{(s_1,x)} (\Win) \geq p$, or in other words, iff $\Val[M]{s_1,x}\geq p$ for every $x \in [0,1/4]$. To recognize whether this is the case, it suffices to approximate the value $\WR{s_1,p}$ up to the absolute error $\delta = 1/8$. Furthermore, for the constructed MDP $M$ we have $\log(|r_{\max}|/(\rho-1))=\log(\mathrm{poly}(w_{\tot}\cdot n^2))$ for some polynomial $\mathrm{poly}$. Thus, the existence of an algorithm approximating $\WR{s_1,p}$ in time polynomial in $|S|\cdot|A|\cdot\log\left(p_{\min}^{-1}\right)\cdot\log(|r_{\max}|/(\rho-1))\cdot\delta^{-1}$ would imply the existence of an polynomial-time algorithm for Knapsack.

(2.) For the second part of the theorem, the proof is almost the same, we divide all the rewards in $M$ by some sufficiently large number, forcing $r_{\max}$ to be polynomial in the encoding size of the Knapsack instance. We then show 
that in order to decide whether $\WR{s,p}\leq 0$ (and thus, whether the instance admits a solution), it suffices to approximate $\WR{s,p}$ up to the absolute error $\calO(1/w_{\tot})$, where $w_{\tot}$ is the sum of weights over all items in the instance.

Formally, the only difference is in the gain function $F$ of the constructed MDP $M$. We put 
\begin{itemize}
 \item $F(s^-_{i},b)=w_i\cdot\rho^{-2(n-i)}/w_{\tot}$, for all $1\leq i \leq n$;
 \item $F(t_1,b)=1$, $F(t_2,b)=F(t_3,b)=-2$;
 \item $F(s_{n+1},b)=-4(1-W/w_{\tot})\cdot n^2$. This ensures that a run that hits $s_{n+1}$ is winning if and only if the current wealth upon entering $s_{n+1}$ is at least $1-W/w_{\tot}$.
 \item Other rewards are zero.
\end{itemize}
It is again easy to verify that the states $t_1$ and $t_2$ are always winning and losing, respectively, and that for a given strategy $\sigma$ the current wealth upon entering $s_{n+1}$ is equal to $\sum_{i \in \{1,\dots,n\}\setminus I^{+}_{\sigma}} w_i/w_{\tot} + \rho^{2n}\cdot x$, where $x$ is the initial wealth. Now let $x$ be any number from the interval $[0,1/(4w_{\tot})]$ and $\sigma$ be any strategy. Then the conditional probability of winning on condition that $s_{n+1}$ is reached is 1 iff $\sum_{i \in \{1,\dots,n\}\setminus I^{+}_{\sigma}} w_i/w_{\tot} + \rho^{2n}\cdot x \geq (w_{\tot}-W)/w_{\tot}$, and 0 otherwise. Since $\rho^{2n}\cdot x\leq 1/(2w_{\tot})$, this conditional probability is 1 if and only if $\sum_{i \in \{1,\dots,n\}\setminus I^{+}_{\sigma}} w_i/w_{\tot} \geq (w_{\tot}-W-1/2)/w_{\tot}$, i.e. iff $I^{+}_{\sigma}$ contains items of weight at most $W$ (since the item weights are integers). Now we can argue in exactly the same way as in the previous part, that $I^{\sigma}_{+}$ is a 
solution to the Knapsack instance iff $\sigma$ ensures winning with probability at 
least $p=1-1/n + V$ from $(s_1,x)$. It follows that the Knapsack instance has a solution if and only if it is possible to win with probability at least $p$ from $s_1$ with any initial wealth between $0$ and $1/(4w_{\tot})$. To check whether this is the case it suffices to approximate $\WR{s_1,p}$ up to the absolute error $\delta=1/(8{w_{\tot}})$. The constructed MDP $M$ has $\log(|r_{\max}|/(\rho-1))=\log(\mathrm{poly}(n))$ for some polynomial $\mathrm{poly}$ (since $r_{\max}\leq 4n^2$). Thus, an approximation algorithm running in time polynomial in $|S|\cdot|A|\cdot\log\left(p_{\min}^{-1}\right)\cdot\log(\delta^{-1})$ and exponential in $\log(|r_{\max}|/(\rho-1))$
would yield a polynomial-time algorithm for knapsack.
\end{proof}

\bibliographystyleadd{abbrv}
\bibliographyadd{references}

\end{document}